\documentclass[10pt]{article}
\usepackage[T1]{fontenc}
\usepackage[utf8]{inputenc}
\usepackage[a4paper,width=150mm,top=25mm,bottom=25mm]{geometry}
\usepackage{macros}

\newcommand{\cF}{\mathcal{F}}

\newcommand{\mcL}{\mathscr{L}}
\newcommand{\mcM}{\mathscr{M}}
\newcommand{\mm}{\mathbf{m}}
\newcommand{\ind}{\mathbf{1}}
\newcommand{\PROD}{\mathsf{prod}}
\newcommand{\infl}{\Psi}
\newcommand{\eps}{\varepsilon}
\renewcommand{\tilde}{\widetilde}

\usepackage{quoting,xparse}

\bibliography{ssm}

\title{Counterexamples to a Weitz-Style Reduction for Multispin Systems}
\author{Kuikui Liu\thanks{email: \texttt{liukui@mit.edu}, 
Massachusetts Institute of Technology} \and Nitya Mani\thanks{Email: \texttt{nmani@mit.edu}, 
Massachusetts Institute of Technology. Mani was supported by the Hertz Graduate Fellowship and the NSF Graduate Research Fellowship Program.} \and Francisco Pernice\thanks{email: \texttt{fpernice@mit.edu}, 
Massachusetts Institute of Technology}}
\date{February 2025}

\begin{document}
\maketitle
\begin{abstract}
In a seminal paper, Weitz \cite{Wei06} showed that for two-state spin systems, such as the Ising and hardcore models from statistical physics, correlation decay on trees implies correlation decay on arbitrary graphs. The key gadget in Weitz's reduction has been instrumental in recent advances in approximate counting and sampling, from analysis of local Markov chains like Glauber dynamics to the design of deterministic algorithms for estimating the partition function. A longstanding open problem in the field has been to find such a reduction for more general multispin systems like the uniform distribution over proper colorings of a graph.

In this paper, we show that for a rich class of multispin systems, including the ferromagnetic Potts model, there are fundamental obstacles to extending Weitz's reduction to the multispin setting. A central component of our investigation is establishing nonconvexity of the image of the belief propagation functional, the standard tool for analyzing spin systems on trees. On the other hand, we provide evidence of convexity for the antiferromagnetic Potts model.

\end{abstract}

\section{Introduction}
Given a graph $G=(V,E)$ and an alphabet of size $q$, a $q$-\emph{spin system} (or \emph{pairwise Markov random field}) is a probability distribution over configurations $\sigma \in [q]^V$ where the probabilities are determined by the pairwise interactions of states $\sigma(u)$, $\sigma(v)$ for neighboring vertices $(u,v)\in E$ (see \cref{defn:mrf}). Important examples include the Ising model from statistical physics, and random solutions to well-studied constraint satisfaction problems like independent sets and proper colorings in graphs. Given such a distribution $\mu$, the following three algorithmic tasks, which are essentially equivalent \cite{JVV86}, are of fundamental importance in statistics and computer science:
\begin{tasks}
    \item\label{task:sample} \textbf{Sampling:} Generate a random $\sigma \in [q]^{d}$ such that $\Law(\sigma) \approx \mu$. %$\norm{\Law(\sigma) - \mu}_{\TV} \leq \epsilon$.
    \item\label{task:count} \textbf{Counting:} Estimate the normalizing constant for the distribution $\mu$, i.e. the \emph{partition function} (see \cref{defn:mrf}).
    \item\label{task:marginal} \textbf{Marginalization:} Given a vertex $v \in V$, estimate the marginal law of $\sigma(v)$ for $\sigma \sim \mu$.
\end{tasks}
One of the main problems in the field of counting and sampling has been to determine for a given $q$-spin system whether the algorithmic tasks above are computationally tractable. In this direction, a powerful intuition coming from statistical physics has been that such questions should be reducible, at least in many important cases, to analogous questions on trees which are typically easier to analyze. This has led to the following meta-question, which has been pursued by many works \cite{Wei06, GK07, BG08, Sly10, LLY13, SST14, GKM15}:
\begin{metaquestion}\label{metaQ:main}
    When can one reduce algorithmic questions about \cref{task:sample,task:count,task:marginal} on general graphs to structural questions about the corresponding distributions on trees?
\end{metaquestion}

The general idea above is best illustrated with an example. Given a \emph{fugacity} parameter $\lambda > 0$, the \emph{hardcore model} on a graph $G=(V,E)$ is the probability distribution $\mu = \mu_{G,\lambda}$ on $\{0,1\}^V$ (or equivalently, subsets of $V$) defined by
\begin{align*}
    \mu(I)&\propto \lambda^{|I|},
\end{align*}
for all independent sets $I\subseteq V$. One of the early seminal results in this area \cite{Wei06, Sly10} showed that \cref{task:sample,task:count,task:marginal} are computationally tractable for the hardcore model on the class of graphs of maximum degree $\Delta$ if and only if the hardcore model on the infinite $\Delta$-regular tree exhibits a certain structural property called \emph{correlation decay} (see \cref{defn:ssm} for a precise statement). An essential part of the algorithmic side of this program was a result due to Weitz \cite{Wei06}, which reduced the computation of marginal laws of vertices in general graphs to analogous computations on trees. On a tree, these marginal laws can be computed exactly via a recursive algorithm known as \emph{belief propagation} (or \emph{tree recursion}), which can be analyzed directly. Weitz also suggested his construction could be extended to the class of all $2$-spin systems, i.e. all pairwise Markov random fields where each vertex can be in one of at most two possible states. This extension was later used by \cite{LLY13} to obtain algorithmic results for a broad class of $2$-spin systems. It has been an outstanding open problem in the field to find an analogous reduction for $q$-spin systems when $q>2$.

While 2-spin theory is by now relatively well-developed, there is a scarcity of techniques in the multispin ($q>2$) setting. The following major conjecture makes this particularly evident.
\begin{conjecture}[Informal; Jerrum \cite{Jer95}]\label{c:jerrum}
Fix $\Delta,q \in \N$. If $q \geq \Delta + 1$, then \cref{task:sample,task:count,task:marginal} admit polynomial-time algorithms for the uniform distribution over proper $q$-colorings of any graph of maximum degree $\Delta$.
%    Fix $\Delta, q \in \mathbb{N}$. Consider the problem of producing an approximate sample from the uniform distribution over proper $q$-colorings of a graph of maximum degree $\Delta$. This problem can be solved in polynomial time if and only if $q\geq \Delta + 1.$
\end{conjecture}
It is known that if $q \leq \Delta$ (and $q$ is even), no efficient algorithm for approximately counting and sampling proper $q$-colorings exists unless $\NP = \RP$ \cite{GSV15}. However, the algorithmic side of this question remains open. As in the hardcore model, the conjectured threshold $\Delta + 1$ for the number of colors is the critical value of $q$ above which the uniform distribution over proper $q$-colorings exhibits correlation decay in the infinite $\Delta$-regular tree. This motivates a natural two-step approach to the conjecture:
\begin{steps}
    \item\label{step:reduce} Reduce the problem to proving (a sufficiently strong version of) correlation decay on trees.
    \item\label{step:decay-trees} Establish this correlation decay property on trees by analyzing belief propagation.
\end{steps}
The second part of this program was nearly resolved in recent work \cite{CLMM23} after earlier, weaker results due to \cite{Jon02, EGHSV19}. More specifically, a strong version of correlation decay called \emph{strong spatial mixing} (see \cref{defn:ssm}) was established on trees of maximum degree $\Delta$ whenever $q \geq \Delta + 3$. However, efficient algorithms for sampling proper $q$-colorings on general graphs of maximum degree $\Delta$ are only known for $q\geq 1.809\Delta$ \cite{Jer95, Vig00, CDMPP19,CV24}. This discrepancy is due to the absence of a suitable reduction for general $q$-spin systems that realizes the first step of the above approach. In this work, we show that there are fundamental barriers to building such a reduction using existing techniques. Our results can be interpreted as the following:
\begin{quoting}[indentfirst=true, leftmargin=\parindent, rightmargin=\parindent]\itshape
For general $q$-spin systems with $q > 2$, there is no Weitz-style gadget that reduces computation of vertex marginals in general graphs to trees.
%There is no Weitz-style reduction to compute the marginal laws of vertices from general graphs to trees for general $q$-spin systems if $q>2$.
\end{quoting}
Our results suggest that a genuinely new idea is required to resolve \cref{metaQ:main} for multispin systems. Let us now make more precise what we mean.

\subsection{Background and Main Results}\label{sec:2.1}
We begin by defining the relevant objects, reinterpreting Weitz's reduction, and formally stating our results.
\paragraph{Notation.} Throughout the rest of the paper, we let $G=(V,E)$ be a graph of maximum degree $\Delta$, and we use $d$ to refer to the degree of a specific vertex in question; hence $d\leq \Delta$. If $v \in V$, we let $G-v$ denote the graph where vertex $v$ has been removed together with all its incident edges. We will write $[q]$ for the space of colors (or spins), and use letters $\mfa,\mfb,\mfc$ to denote elements of $[q]$. The space of probability distributions over $[q]$ will be denoted $\triangle_{q}$. If $\mfc \in [q]$ is a color, we let $\delta_\mfc \in \triangle_{q}$ be the probability distribution on $[q]$ that is a point mass at $\mfc$. We let $\ind$ denote the all ones vector.

Let $\mu$ be a probability distribution over $[q]^{\mathcal{U}}$ for a finite set $\mathcal{U}$. If $S \subseteq \mathcal{U}$, we let $\mu_{S}$ be the marginal law of $\sigma_{S} \defeq \wrapp{\sigma(u)}_{u \in S}$ where $\sigma \sim \mu$; if $i \in \mathcal{U}$, we write $\mu_{i}$ instead of $\mu_{\{i\}}$. A \emph{pinning} is a partial assignment $\tau:\Lambda \to [q]$ on $\Lambda\subseteq \mathcal{U}$ such that $\mu_{\Lambda}(\tau)>0$; we let $\mu^\tau$ be the law of $\sigma \sim \mu$ conditioned on $\sigma_\Lambda = \tau$. If $\mu,\nu$ are distributions on state spaces $\Omega,\Sigma$, respectively, then we write $\mu \otimes \nu$ for the product distribution over $\Omega \times \Sigma$ given by $(\mu \otimes \nu)(\omega,\sigma) = \mu(\omega)\nu(\sigma)$ for all $(\omega,\sigma) \in \Omega \times \Sigma$. Finally, if $\mu$ is a probability distribution on $[q]^\mathcal{U}$ and $\bm{\lambda} \in \R^{\mathcal{U}\times [q]}_{\geq 0}$ is an \emph{external field}, we let $\bm{\lambda} \star \mu$ be the distribution on $[q]^{\mathcal{U}}$ defined as follows:
\begin{align*}
    (\bm{\lambda} \star \mu)(\sigma) &\propto \mu(\sigma) \cdot \prod_{i \in \mathcal{U}} \bm{\lambda}_{i, \sigma(i)}, \qquad \forall \sigma \in [q]^{\mathcal{U}}.
\end{align*}
Since we are working with graphical models, depending on the context, we will typically use $\mathcal{U} = V$ or $\mathcal{U} = [d]$, the latter of which should be thought of as the neighborhood of a vertex. If we wish to highlight the dependence of $\mu$ on $G$, we will write $\mu=\mu_G$.
\begin{definition}\label{defn:mrf}
    Given a graph $G = (V,E)$, a number of colors $q$, a symmetric \emph{interaction matrix} $A\in \R^{q\times q}_{\geq0}$, we define the Gibbs distribution of the associated \emph{$q$-spin system} as the following probability distribution $\mu$ on $[q]^V$:
    \begin{align*}
        \mu(\sigma) &= \frac{1}{Z} \prod_{uv \in E}A_{\sigma(u), \sigma(v)}
    \end{align*}
    where $Z$ is the normalizing constant guaranteeing $\sum_{\sigma} \mu(\sigma) = 1$; we call $Z$ the \emph{partition function}.
\end{definition}
Oftentimes, an external field will also be applied to these distributions. We remark that the hardcore model over independent sets and the uniform distribution over proper colorings are examples of spin systems, with interaction matrices $A = \begin{bmatrix} 0 & 1 \\ 1 & 1 \end{bmatrix}$ and $A = \ind\ind^{\top} - I$, respectively. %(respectively) the following parameters:
%\begin{align*}
%    q=2,\qquad A = \begin{bmatrix}
%    1 & 1 \\
%    1 & 0
%\end{bmatrix}, 
%\qquad \bm{\lambda}_u = (1,  \lambda) \;\; \forall u\in V
%\end{align*}
%and
%\begin{align*}
%    A = \ind\ind^\top - I_{q\times q}, 
%\qquad \bm{\lambda}_u = \ind \;\;\forall u\in V.
%\end{align*}

Given a general $q$-spin system on a graph $G$ and a vertex $v$ with neighborhood $N(v)$, one can compute the marginal law $\mu_v$ as a function of the joint law $\mu_{N(v)}$ in the graph where $v$ has been removed.
The following is a well-known generalization of the so-called belief propagation equations to general graphs (see, e.g. \cite[Proposition 6]{GK12}).
\begin{lemma}[Recursion for General Graphs]\label{lem:recursion-general}
For every $q \in \N$, symmetric interaction matrix $A \in \R_{\geq0}^{q \times q}$, graph $G=(V,E)$, and vertex $v \in V$ with $d$ neighbors $N(v) = \{u_{1},\dots,u_{d}\}$, we have
\begin{align*}
    \mu_{G,v} = F_{A,d}\wrapp{\mu_{G-v, N(r)}},
\end{align*}
where the function $F_{A,d}$ is defined as follows: for an arbitrary distribution $\nu$ over $[q]^{d}$, we have
\begin{align}
    F_{A,d,\mfc}(\nu) &\defeq \frac{G_{\mfc}(\nu)}{\sum_{\mfb \in [q]} G_{\mfb}(\nu)}, \qquad \forall \mfc \in [q] \label{eq:recursion-general} \\
    G_{A,d,\mfc}(\nu) &\defeq \E_{\tau \sim \nu}\wrapb{\prod_{i=1}^{d} A_{\mfc, \tau(i)}}, \qquad \forall \mfc \in [q]. \label{eq:unnorm-recursion-general}
\end{align}
\end{lemma}
\begin{remark}
One can also establish an analogous recursion in the presence of external fields. Throughout this paper, we drop the subscripts $A,d$ when they are clear from context. In particular, even though belief propagation for the interaction matrix $A$ is formally described by the \emph{family} of functions $\{F_{A,d}\}_{d \in \N}$, we sometimes abuse notation and simply write $F$. We apply the same convention to $G$.
\end{remark}
To keep this paper self-contained, a short proof of \cref{lem:recursion-general} is provided in \cref{app:BP-general}. Note that when $T$ is a tree rooted at $r$, with corresponding subtrees $T_{i}$ rooted at $u_{i}$ for each $i = 1,\dots,d$, then $\mu_{G-r,N(r)}$ factorizes as the product measure $\bigotimes_{i=1}^{d} \mu_{T_{i},u_{i}}$. In particular, restricting \cref{eq:recursion-general} to product measures $\bigotimes_{i=1}^{d} \nu_{i}$ recovers the more familiar form
\begin{equation}\label{e:f-product}
    F_{\mfc}(\nu_{1},\dots,\nu_{d}) = \frac{\prod_{i=1}^{d} \sum_{\mfb \in [q]} A_{\mfc,\mfb} \nu_{i}(\mfb)}{\sum_{\mfa \in [q]} \prod_{i=1}^{d} \sum_{\mfb \in [q]} A_{\mfc,\mfa} \nu_{i}(\mfa)}
\end{equation}
of the belief propagation functional. It is well-known that correlation decay on trees (see \cref{step:decay-trees}) can be readily deduced from contraction of $F$ in the space of product measures with respect to a suitable (pseudo)metric. In terms of $F$, the goal of obtaining a Weitz-style reduction for \cref{step:reduce} can be viewed as reducing computation of $F$ on a general measure $\mu$ to computing $F$ on product measures whose marginals are obtained from $\mu$ by conditioning and marginalization. This leads to the following abstract reintepretation of Weitz's general reduction.

\begin{proposition}[\cite{Wei06}]\label{prop:weitz-reinterpret}
Let $A \in \R_{\geq0}^{2 \times 2}$ be an arbitrary $2$-spin interaction matrix with spin set $\{0, 1\}$. For every $d \in \N$, there exists a single \emph{universal} collection of external fields $\bm{\lambda}^{(1)},\dots,\bm{\lambda}^{(d)} \in \R_{\geq0}^{[d] \times \{0,1\}}$ such that for every distribution $\mu$ on $\{0,1\}^{d}$, we have that
\begin{align*}
    F(\mu) = F\wrapp{\bigotimes_{i=1}^{d} \wrapp{\bm{\lambda}^{(i)} \star \mu}_{i}}.
\end{align*}
\end{proposition}

We postpone the proof to \cref{app:BP-general}, which essentially just follows \cite{Wei06}. \cref{prop:weitz-reinterpret} is the crucial gadget which allows one to reduce the existence of efficient algorithms for general graphs to correlation decay on trees for $2$-spin systems like the hardcore and Ising models. In view of \cref{prop:weitz-reinterpret}, it's natural to ask whether the same statement holds for general $q$-spin systems with $q > 2$. Below, we consider an even weaker version of this question, to which an affirmative answer would nonetheless yield efficient algorithms, and in particular a resolution of \cref{c:jerrum}. A positive answer to the following question (in the $\varphi=\id$ case) could be interpreted as the existence of a \emph{randomized} construction of a Weitz-style computation tree to calculate the marginal law of a given vertex in a general graph for general $q$-spin systems. We relax our question to allow potential functions $\varphi$ apart from the identity (e.g. $\varphi = \log$ or $\varphi = \sqrt{\cdot}$). As we show in \cref{sec:main-conj-implies-SSM,sec:main-conj-implies-SI}, one can obtain many of the same implications as from~\cref{prop:weitz-reinterpret} by applying the \emph{potential method}, a widely used tool (see, e.g., \cite{GK07, LLY13, RSTVY13}).

\begin{question}\label{q:main}
Is it the case that for every pair of positive integers $q \ge 3, d \ge 2$ and every $q$-spin interaction matrix $A \in \RR^{q \times q}_{\geq 0}$, there exists a univariate monotone potential function $\varphi: [0,1] \to \RR$ and a \emph{universal} distribution $\xi$ on $d$-tuples of external fields $(\bm{\lambda}^{(1)}, \ldots, \bm{\lambda}^{(d)})$, where $\bm{\lambda}^{(i)} \in \R_{\geq0}^{d \times q}$ for each $i\in [d]$, such that for every distribution $\mu$ on $[q]^d$ we have that 
\begin{equation}\label{e:round-to-product}
\varphi(F_A(\mu)) = \mathbb{E}_{(\bm{\lambda}^{(1)}, \ldots, \bm{\lambda}^{(d)}) \sim \xi} \left[ \varphi \left( F_A \left( \bigotimes_{i = 1}^d \wrapp{\bm{\lambda}^{(i)} \star \mu}_i \right) \right) \right],
\end{equation}
where $\wrapp{\bm{\lambda}^{(i)} \star \mu}_i$ is obtained by considering the ``tilted'' $\bm{\lambda}^{(i)} \star \mu$ and marginalizing out $[d] \backslash \{i\}$?
\end{question}
The condition that the potential $\varphi$ must be monotone is essentially without loss of generality, since $\varphi$ must be invertible and it is natural to assume that it is continuous. Moreover, all previously used potentials in the study of $q$-spin systems have been monotone, including $\varphi = \log$ and $\varphi = \sqrt{\cdot}$~\cite{GKM15,GK12,LLY12,LLY13,CLMM23}. % Monotonicity also facilitates a short reduction from~\cref{q:main} to strong spatial mixing and $O(1)$-spectral independence, as elaborated upon in~\cref{app:main-q-consequences}.
In~\cref{lem:image-F-cvxhull} below, we show that \cref{e:round-to-product} is always achievable with $\varphi = \id$ if the mixture measure $\xi$ is allowed to depend on the input measure $\mu$. If $\xi$ could be made universal as in \cref{q:main} (i.e. not depending on the input measure $\mu$), it would imply that one can reduce strong spatial mixing from general graphs to trees. It would also imply spectral independence and hence optimal mixing of Glauber dynamics on general graphs whenever belief propagation is a contraction on trees. We discuss these implications in detail in \cref{app:main-q-consequences}.

In this work, we construct a wide class of spin systems for arbitrary $q \ge 3, d \ge 2$ where the answer to~\cref{q:main} is negative. 
\begin{theorem}\label{t:main}
For every pair of positive integers $q \ge 3, d \ge 2$, there exists a family of symmetric interaction matrices where the answer to~\cref{q:main} is negative for any monotone potential $\varphi$.
% any symmetric interaction matrix $A \in \RR^{q \times q}_{\geq0}$, and any $\beta \ge 1$, define new symmetric interaction matrix  
% $$B = \begin{bmatrix}
% \beta & \ind_q^\top \\
% \ind_q & A 
% \end{bmatrix} \in \RR_{\geq0}^{(q+1) \times (q+1)}.$$
\end{theorem}

In \cref{sec:counterexample}, we explicitly describe a class of spin systems for which \cref{q:main} has a negative answer. Our counterexamples include the ferromagnetic Potts model and a natural model of partial coloring in a $q$-spin system, with a special $(q+1)$'st non-interacting spin $\mfc^*$. We give a technical description of this class of spin systems in~\cref{t:technical}. As we show in~\cref{prop:signature}, our class of counterexamples includes $(q + 1) \times (q + 1)$ interaction matrices with $k$ positive eigenvalues for any $2 \le k \le q+1$.

To certify a negative answer to \cref{q:main} for a specific interaction matrix $A$, we construct a pair of probability distributions $\mu,\nu$ for which there is no such $\xi$ which satisfies \cref{e:round-to-product} for both distributions $\mu,\nu$ simultaneously. We highlight that the witnesses $\mu,\nu$ we construct are not contrived. They are quite reasonable from the perspective of marginal bounds, which are the most widely-used devices to restrict the domain of $F$ under consideration. While they also exhibit very strong correlations, in some sense the point of reductions like \cref{prop:weitz-reinterpret} and \cref{q:main} is to tame strongly correlated distributions. Indeed, in general graphs, the vertices of the neighborhood $N(r)$ of a vertex $r$ can be close to one another even after deleting $r$.

% \subsection{The image of product measures under \textsf{BP}}
% \nitya{Where should this stuff be? I'm not convinced here is correct.} \fran{I was thinking of including this at the end of section 1.1 (without it having its own section 1.2).}

\paragraph{On the Image of Belief Propagation}
Our results also have consequences for the structure of the image of the set of product measures under $F$.
\begin{definition}\label{defn:calF-and-calFprod}
For integers $d,q \geq 2$ and a symmetric interaction matrix $A \in \R_{\geq0}^{q \times q}$, we define  
\begin{align}
    \cF_{A,d}^{\PROD} \defeq \left\{ F_A(\nu) : \nu \text{ is a product measure over } [q]^d \right\} \subseteq \triangle_q, \label{eq:F-prod} \\
    \cF_{A,d} \defeq \left\{ F_A(\mu) : \mu \text{ is an arbitrary measure over } [q]^d \right\} \subseteq \triangle_q. \label{eq:F-arbitrary}
\end{align}
\end{definition}
As we show in~\cref{lem:image-F-cvxhull}, these two sets are intimately related, with the latter being the convex hull of the former, i.e. $ \cF_{A,d} = \conv\wrapp{\cF_{A,d}^{\PROD}}$. These families of measures $\cF_{A,d}^{\PROD}, \cF_{A,d}$ naturally appear in our study of Weitz-style reductions for $q$-spin systems, with the measures appearing on the right hand side of~\cref{e:round-to-product} lying in $\cF_{A, d}^{\PROD}$. A key ingredient in our main result,~\cref{t:technical}, is to certify that for $B$ as in~\cref{e:b}, the set $\cF_{B,d}^{\PROD}$ is not convex and thus for such spin systems, $\cF_{B,d} \supsetneq \cF_{B,d}^{\PROD}$; see \cref{l:l1,rmk:nonconvex}. In particular, there are distributions $\mu$ on $[q]^d$ where $F_B(\mu)$ cannot be achieved as $F_B(\nu)$ by any single product measure $\nu$. This lack of convexity is crucial for \cref{t:main}.

The interaction matrices we construct to prove \cref{t:main} all have at least two positive eigenvalues. In other words, the presence of ferromagnetic interactions obstructs convexity of the set $\mathcal{F}_{A,d}^{\PROD}$ and yields a negative answer to \cref{q:main}. In light of this, it is natural to wonder if repulsion restores convexity of $\mathcal{F}_{A,d}^{\PROD}$ (i.e. equality of sets $\mathcal{F}_{A,d} = \mathcal{F}_{A,d}^{\PROD}$), which could then lead to a positive answer to \cref{q:main} for antiferromagnetic models like colorings. We present evidence that this is the case by showing that for the finite-temperature \emph{antiferromagnetic Potts model}, ``most'' distributions $\mu$ on $[q]^{d}$ one would encounter ``in the wild'' are such that $F(\mu) = F(\nu)$ for some product measure $\nu$. One can interpret this as saying that for the antiferromagnetic Potts interaction matrix $A = \allone\allone^{\top} - (1 - \beta)I$ where $0 \leq \beta \leq 1$, the set $\mathcal{F}_{A,d}^{\mathsf{prod}}$ at least has a ``large convex bulk''. For convenience, we will work in the range of temperatures $\beta$ in which the resulting Gibbs distribution exhibits correlation decay on the infinite $(d+1)$-regular tree, although our result below can be made to accommodate a larger range of temperatures.
\begin{theorem}\label{thm:antiferro-potts-cvx-bulk}
Fix an integer $q \geq 2$, and suppose $d \geq 2q$. Consider the antiferromagnetic Potts interaction matrix $A = \allone_{q}\allone_{q}^{\top} - (1 - \beta)I_{q}$ in the uniqueness regime with respect to the infinite $(d+1)$-regular tree, i.e. $\max\wrapc{0, 1 - \frac{q}{d + 1}} \leq \beta \leq 1$. If $\mu$ is any probability distribution over $[q]^{d}$ satisfying the lower tail bounds
\begin{align}\label{eq:color-density-lower-tail}
    \forall \mfc \in [q], \qquad \Pr_{\tau \sim \mu}\wrapb{\#\{i \in [d] : \tau(i) = \mfc\} \leq \gamma_{*} \cdot \frac{d}{q}} \leq \epsilon,
\end{align}
where $\gamma_{*} = \min\wrapc{1 - \sqrt{\beta} + \epsilon, \frac{1}{10}}$ and $\epsilon \leq \frac{1}{4}$ is arbitrary, then there exists a product measure $\nu$ over $[q]^{d}$ such that $F(\mu) = F(\nu)$.
\end{theorem}
One should interpret $\gamma_{*}$ as an ``effective lower bound'' on the number of neighbors colored $\mfc$ as a fraction of the completely balanced allocation $d/q$, for any given color $\mfc \in [q]$. Also, note that $1 - \sqrt{\beta} \leq 1 - \beta \leq \frac{q}{d+1}$, which should be thought of as small since we fix $q$ and let $d \geq 2q$ grow.

\begin{remark}
We emphasize that these tail bounds are quite weak. To illustrate the applicability of \cref{thm:antiferro-potts-cvx-bulk}, observe that if $\mu = \mu_{G-r,N(r)}$ is the marginal distribution over the neighborhood of a vertex $r$ in a graph $G$ of maximum degree $d + 1$, then $\mu$ must satisfy the marginal lower bound
\begin{align*}
    \mu_{i}(\mfc) \geq \frac{\beta^{d}}{q - 1 + \beta^{d}}
\end{align*}
for each neighbor $u_{i} \in N(r)$, which is attained by pinning the entire neighborhood of $u_{i}$ to the color $\mfc$. Hence, the law of $\alpha(\mfc) \defeq \#\{i \in [d] : \tau(i) = \mfc\}$ stochastically dominates the random variable $\sum_{i=1}^{d} X_{i}$, where $X_{i} = 1$ with probability $\frac{\beta^{d}}{q - 1 + \beta^{d}}$ and $X_{i} = 0$ with probability $\frac{q - 1}{q - 1 + \beta^{d}}$. In particular, by taking $\eps = c/ d$ for some sufficiently small $c < 1$, we have that $$\gamma_* \le 1 - \sqrt{\beta} + \eps \le 1 - \beta + \eps \le \frac{q}{d+1} + \frac{c}{d},$$
so that $\gamma_{*} \cdot \frac{d}{q} < 1$. %, which is possible since $1 - \sqrt{\beta} \leq 1 - \beta \le \frac{q}{d+1}$ and we can take $\epsilon = c/d$ with a sufficiently small leading constant $c < 1$, then
This yields
\begin{align*}
    \Pr_{\tau \sim \mu}\wrapb{\alpha(\mfc) \leq \gamma_{*} \cdot \frac{d}{q}} = \Pr_{\tau \sim \mu}\wrapb{\alpha(\mfc) = 0} \leq \wrapp{1 - \frac{\beta^{d}}{q - 1 + \beta^{d}}}^{d}.
\end{align*}
Since $q$ is fixed and $\beta \geq 1 - \frac{q}{d + 1}$, we have $1 - \frac{\beta^{d}}{q - 1 + \beta^{d}} \leq c_{q}$ for a constant $0 < c_{q} < 1$ independent of $d$, and so we certainly have $\Pr_{\tau \sim \mu}\wrapb{\alpha(\mfc) = 0} < \epsilon$ even for moderate $d$ (and certainly for large $d$).
\end{remark}

\subsection{Prior Works}
As mentioned in the introduction, the relation between the belief propagation equations and complexity-theoretic questions for general graphs has been heavily studied in the literature. After Weitz's seminal paper \cite{Wei06}, Sly \cite{Sly08} gave an example of a $q$-spin system where $q > 2$ for which a weak form of correlation decay known as \emph{weak spatial mixing} holds in the infinite $\Delta$-regular tree, but does not hold on some $\Delta$-regular graphs. We remark that a Weitz-style reduction would show (as it does in the $2$-spin case) that \emph{strong spatial mixing} (see \cref{defn:ssm}) on the infinite $\Delta$-regular tree implies strong spatial mixing on all $\Delta$-regular graphs. Hence, as Sly points out in his paper, his counterexample does not already rule out a Weitz-style reduction, because both the assumption and the conclusion of such a reduction would be stronger than what his example considers. For this reason, the results of \cite{Sly08} are incomparable to ours.

Li, Lu and Yin \cite{LLY13}, in the course of giving a characterization of all antiferromagnetic 2-spin systems that exhibit strong spatial mixing, gave an example of a 2-spin system for which strong spatial mixing on the $\Delta$-regular tree is not monotone in $\Delta$. In other words, there are 2-spin systems for which strong spatial mixing holds in the $\Delta$-regular tree but fails in the $k$-regular tree for some $k<\Delta.$ Sinclair, Srivastava and Thurley \cite{SST14} proved that, in the antiferromagnetic Ising model, correlation decay implies strong spatial mixing, and showed that the monotonicity in $\Delta$ for 2-spin systems is recovered under a different parametrization than the one used by \cite{LLY13}.

Important works \cite{Sly10, SS14, GSV15, GSV16, GSVY16} have also shown that failure of strong spatial mixing on trees successfully predicts algorithmic hardness for counting and sampling in antiferromagnetic spin systems. This, combined with known algorithmic results, suggests a tantalizing picture which has been confirmed in a number of important cases: for antiferromagnetic systems, algorithmic tractability of \cref{task:sample,task:count,task:marginal} in general graphs is \emph{equivalent} to strong spatial mixing on trees.

In \cite{CLMM23}, in addition to establishing strong spatial mixing for colorings on trees down to $q\geq \Delta + 3$ for any $\Delta\geq 3$, the authors gave a general reduction from rapid mixing of Glauber dynamics on large-but-constant girth graphs to strong spatial mixing on trees. This reduction is via a recursive coupling argument and is of a different nature to Weitz's combinatorial construction.

\paragraph{On Previous Computation Trees for Multispin Systems} Gamarnik and Katz \cite{GK07}, building on earlier ideas of Gamarnik and Bandyopadhyay \cite{BG08}, developed an extension of Weitz's combinatorial reduction that works for arbitrary $q$-spin systems and any $q \geq 2$, and applied it to proper colorings. However, crucially, their construction expresses the marginal law of a vertex in a general graph of maximum degree $\Delta$ as the result of a recursion analogous to \cref{eq:recursion-general} \emph{on a computation tree of maximum degree} $\Delta \cdot q$. This blowing up of the degree prevents a direct reduction from strong spatial mixing on general graphs to trees, but it still enables one to build deterministic approximate counting algorithms which work for $q \geq 2.58071\Delta + 1$ \cite{GK07, LY13}; see also \cite{LSS20, BBR24} for recent progress on deterministic counting algorithms using related ideas.

Gamarnik, Katz and Misra \cite{GKM15} used this construction to derive improved bounds for strong spatial mixing for random colorings on triangle-free graphs. One major limitation of this approach is that in order to derive algorithmic consequences, it does not suffice to establish contraction of belief propagation on trees of maximum degree $\Delta$ under any metric. Instead, one needs to show contraction \emph{with respect to an ``$\ell_\infty$-type'' metric}. On the other hand, for instance in the case of random colorings, the best known strong spatial mixing results on trees \cite{CLMM23} crucially rely on an $\ell_2$-type norm.

% \section{Preliminaries}

% \subsection{Notation}
% Fix $q \geq 2$ and let $A \in \R_{\geq0}^{q \times q}$ be a symmetric interaction matrix. For a graph $G=(V,E)$ and a collection of vertex-dependent external fields $\bm{\lambda} = \{\bm{\lambda}_{v}\}_{v \in V}$, where $\bm{\lambda}_{v} \in \R_{\geq0}^{q}$ for all $v \in V$, the \emph{Gibbs distribution} over $[q]^{V}$ of the associated \emph{spin system} is given by
% \begin{align}\label{eq:gibbs-dist}
%     \mu_{G}(\sigma) = \mu_{G,A,\bm{\lambda}}(\sigma) \propto \prod_{uv \in E} A_{\sigma(u),\sigma(v)} \prod_{v \in V} \bm{\lambda}_{v,\sigma(v)}, \qquad \forall \sigma : V \to [q].
% \end{align}
% The associated normalizing constant, i.e. the \emph{partition function}, is precisely the quantity
% \begin{align}
%     Z_{G} = Z_{G,A}(\bm{\lambda}) = \sum_{\sigma : V \to [q]} \prod_{uv \in E} A_{\sigma(u),\sigma(v)} \prod_{v \in V} \bm{\lambda}_{v,\sigma(v)}.
% \end{align}

% \subsection{\textsf{BP} on general graphs}
% A classical result in the theory of approximate counting and sampling is that approximately computing $Z_{G}$ is equivalent to approximately sampling from the Gibbs distribution $\mu_{G}$ \cite{JVV??}. Moreover, computing the partition function $Z_{G}$ can be reduced to estimating the \emph{marginal distributions} $\mu_{G,r}^{\tau}$ over $[q]$ of each vertex $r$, perhaps conditioned on various boundary conditions $\tau : S \to [q]$ where $S \subseteq V \setminus \{r\}$. Belief propagation provides a recursive way of computing these marginal distribution in trees. 

% 

\section{Proof of \texorpdfstring{\cref{t:main}}{MainThm}}\label{sec:counterexample}
We establish the following precise version of~\cref{t:main}, answering~\cref{q:main} in the negative.
\begin{defn}\label{d:m-ad}
Given positive integers $q \ge 3, d \ge 2$ and a symmetric interaction matrix $A \in \RR_{\geq0}^{q \times q}$ with spin set $[q]$, let 
$$\mcM(d, A) \defeq \arg\max_{\mfb \in [q]} \sum_{\mfc \in [q]} A_{\mfc, \mfb}^d.$$
\end{defn}

\begin{theorem}\label{t:technical}
Fix positive integers $q \ge 2, d \ge 2$. Let $A \in \RR^{q \times q}_{\geq0}$ be any symmetric interaction matrix with spin set $[q]$ that satisfies the following pair of conditions:
\begin{enumerate}[(a)]
    \item Entrywise, we have $A \ge \ind_{q} \ind_{q}^{\top}$ with strict inequality holding for at least one entry.
    \item We have $|\mcM(d, A)| \ge 2$, and there exists  $\mfc \in [q]$ and $\mfa_{1},\mfa_{2} \in \mcM$ such that $A_{\mfc,\mfa_{1}} \neq A_{\mfc,\mfa_{2}}$.
\end{enumerate}
Then, for any $\beta \ge 1$ and any monotone potential $\varphi$, there does not exist $\xi$ that satisfies~\cref{e:round-to-product} with respect to the symmetric interaction matrix  
\begin{equation}\label{e:b}
B = B(\beta, A) \defeq \begin{bmatrix}
\beta & \ind_q^\top \\
\ind_q & A 
\end{bmatrix} \in \RR_{\geq0}^{(q+1) \times (q+1)}
\end{equation}
for all probability distributions $\mu$ over $\wrapp{[q] \cup \{\mfc^{*}\}}^{d}$, where $[q]$ denotes the $q$ spins originating from interaction matrix $A$, and $\mfc^*$ denotes the special spin introduced in constructing auxiliary interaction matrix $B$.
\end{theorem}

% We will denote by $[q]$ the $q$ spins originating from interaction matrix $A$, and by $\mfc^*$ the special spin introduced in constructing auxiliary interaction matrix $B$, with respect to which we will obtain a negative answer to \cref{q:main}.

A well-known example of such an interaction matrix is the ferromagnetic Potts interaction matrix, where $A = \ind_{q}\ind_{q}^{\top} + (\beta - 1)I_{q}$ and $\beta > 1$. In \cref{prop:signature}, we construct interaction matrices $A$ satisfying the assumptions of \cref{t:technical} which demonstrate that $B$ can have exactly $k$ positive eigenvalues for any $2 \leq k \leq q + 1$. We note that all such matrices necessarily have at least two positive eigenvalues, as we show in \cref{lem:atleasttwo-poseig}. Hence, our examples notably exclude the antiferromagnetic setting, where the interaction matrix has exactly one positive eigenvalue.

Our strategy to prove~\cref{t:technical} will be to use the extremizers of $F_{B,\mfc^*}(\cdot)$ to force any candidate distribution $\xi$ to depend on the input distribution $\mu$ over $[q]^d$. To do this, we characterize those distributions $\nu$ for which $F_{B, \mfc^*}(\nu)$ is maximized.

\begin{lemma} \label{l:l1}%[Marginal upper bounds] 
For an interaction matrix \( A \in \mathbb{R}^{q \times q}_{\geq0} \) satisfying the assumptions of \cref{t:technical}, \( \beta \geq 1 \), and \( B \) as in~\cref{e:b}, we have the bound
\[
F_{B,\mfc^*}(\nu) \geq \frac{1}{1 + \max_{\mfb \in [q]} \sum_{\mfc \in [q]} A^d_{\mfc,\mfb}}
\]
for all product measures \( \nu \) over \( ([q] \cup \{\mfc^*\})^d \). 
Moreover, equality is attained if and only if $\nu = \delta_{\mfb^*}^{\otimes d}$ for some $\mfb^* \in \mcM(d, A)$ as defined in~\cref{d:m-ad}.
\end{lemma}
\begin{remark}\label{rmk:nonconvex}
Already, this lemma certifies that $\mathcal{F}_{B,d}^{\PROD}$ is not a convex set and hence, by \cref{lem:image-F-cvxhull}, $\mathcal{F}_{B,d}^{\PROD} \subsetneq \mathcal{F}_{B,d}$. Indeed, this lemma shows that the subset $\argmin_{\bm{p} \in \mathcal{F}_{B,d}^{\PROD}} \langle \allone_{\mfc^{*}}, \bm{p} \rangle$ is a discrete set with at least two points, which cannot be convex.
\end{remark}

\begin{proof}
Let $B = B(\beta, A)$ as in~\cref{e:b} and $\mcM = \mcM(d, A)$ as in~\cref{d:m-ad}.
Define 
\begin{equation}\label{e:ratio-c}
\widetilde{R}_A(\nu) \defeq \frac{1 - F_{B, \mfc^*}(\nu)}{F_{B, \mfc^*}(\nu)} = \sum_{\mfc \in [q]} \prod_{i=1}^{d} \left(\frac{ \nu_i(\mfc^*) + \sum_{\mfb \in [q]} A_{\mfc,\mfb} \nu_i(\mfb)}{1 + (\beta - 1)\nu_i(\mfc^*)} \right),
\end{equation}
where in the final equality, we apply~\cref{e:f-product}.
Since $A \geq \ind_q \ind_q^\top$ entrywise and 
\( \beta \geq 1 \), \( \widetilde{R}_A(\nu) \) is maximized when \( \nu_i(\mfc^*) = 0 \) for all \( i \in [d] \). Hence,
\[
\inf_{\nu \text{ over } ([q] \cup \{\mfc^*\})^d} F_{B,\mfc^*}(\nu) = \frac{1}{1 + \sup_{\nu\text{ over }[q]^d}  \underbrace{\sum_{\mfc \in [q]} \prod_{i=1}^{d} \sum_{\mfb \in [q]} A_{\mfc,\mfb} \nu_i(\mfb)}_{=: Z_A(\nu)} } = \frac{1}{1 + \sup_{\nu} Z_A(\nu)}.
\]
Thus, to prove the claim we wish to show that for all product distributions $\nu$ over $[q]^d$, 
$Z_A(\nu) \leq \max_{\mfb \in [q]} \sum_{\mfc \in [q]} A^d_{\mfc,\mfb}$, with equality if and only if $\nu = \delta_{\mfb^*}^{\otimes d}$ for some $\mfb^* \in \mcM$.
To see this, observe that
\begin{align*}
Z_A(\nu) &\leq \sup_{\mu \in \triangle_{[q]^d}} \mathbb{E}_{\tau \sim \mu} \left[ \sum_{\mfc \in [q]} \prod_{i=1}^{d} A_{\mfc,\tau(i)} \right] \\
&= \max_{\tau \in [q]^d} \sum_{\mfc \in [q]} \prod_{i=1}^{d} A_{\mfc,\tau(i)} \tag{linearity of \( \mu \mapsto \mathbb{E}_{\tau \sim \mu} \left[ \sum_{c \in [q]} \prod_{i=1}^{d} A_{c,\tau(i)} \right] \)} \\
&\leq \sup_{\alpha \in \triangle_q} \sum_{\mfc \in [q]} \prod_{\mfb \in [q]} A_{\mfc,\mfb}^{\alpha(\mfb) \cdot d} \tag{symmetry of $A$, with $\alpha(\mfb) := \frac{|\{ i \in [d] : \tau(i) = \mfb \}|}{d}$} \\
&\leq \sup_{\alpha \in \triangle_{q}} \sum_{\mfb \in [q]} \alpha(\mfb) \sum_{\mfc \in [q]} A_{\mfc,\mfb}^{d} \tag{AMGM Inequality} \\
&\leq \max_{\mfb \in [q]} \sum_{\mfc \in [q]} A^d_{\mfc,\mfb}. \tag{linearity of $\alpha \mapsto \sum_{\mfb \in [q]} \alpha(\mfb) \sum_{\mfc \in [q]} A_{\mfc,\mfb}^{d}$}
\end{align*}
Our second assumption on $A$ precludes $A$ being proportional to the all-ones matrix $\allone_{q}\allone_{q}^{\top}$, and so the inequalities above all become equalities if and only if \( \nu = \delta_{\mfb^*}^{\otimes d} \) for some $\mfb^* \in \mcM$.
\end{proof}

\begin{lemma}\label{l:expand-fb}
Let \( A \in \mathbb{R}^{q \times q}_{\geq 0} \) be an interaction matrix satisfying the assumptions of~\cref{t:technical}, \( \beta \geq 1 \), and \( B \) as in~\cref{e:b} and $\mcM$ as in~\cref{d:m-ad}. Let \( \zeta \) be a probability distribution over \( \mcM \defeq \mcM(d, A) \), and define the distribution 
\( \mu^{(\zeta)} \) on \( [q]^d \) by
\[
\mu^{(\zeta)}  \defeq \sum_{\mfa \in \mcM} \zeta(\mfa) \cdot \delta_{\mfa}^{\otimes d}.
\]
Then, we have that
\[
F_B\wrapp{\mu^{(\zeta)}} = \sum_{\mfa \in \mcM} \zeta(\mfa) \cdot F_B\wrapp{\delta_{\mfa}^{\otimes d}}.
\]
\end{lemma}
\begin{proof}
Note that 
$$F_{B, \mfc}\wrapp{\mu^{(\zeta)}} = \frac{G_{B, \mfc}\wrapp{\mu^{(\zeta)}}}{\sum_{\mfb \in \{[q] \cup \{\mfc^*\}\}} G_{B, \mfb}\wrapp{\mu^{(\zeta)}}},$$
and that the functions $G_{B, \mfc}(\nu)$ as defined in~\cref{eq:unnorm-recursion-general} are linear in the vector $\nu \in \triangle_{[q]^d}$. Note that for any $\mfa \in \mcM$ the quantity $\sum_{\mfb \in [q] \cup \{\mfc^*\}} G_{B,\mfb}(\delta_{\mfa}^{\otimes d})$ is a constant. Therefore, since $\zeta$ is a probability distribution, the quantity 
$$\sum_{\mfb \in [q] \cup \{\mfc^*\}} G_{B,\mfb}\wrapp{\mu^{(\zeta)}}$$
is a constant independent of $\zeta$. The result then follows by linearity of $G$.
\end{proof}

\begin{proof}[Proof of~\cref{t:technical}]
As before, we take $B = B(\beta, A)$ per~\cref{e:b}.
Suppose for contradiction that there was some monotone potential function $\varphi: [0,1] \to \RR$ such that~\cref{e:round-to-product} held with respect to $B$, so that there is some distribution $\xi$, not depending on $\zeta$, where
\begin{equation}\label{e:contradict}
    \varphi\wrapp{F_B\wrapp{\mu^{(\zeta)}}} = \mathbb{E}_{\wrapp{\bm{\lambda}^{(1)}, \ldots, \bm{\lambda}^{(d)}} \sim \xi} \left[ \varphi \left( F_B \left( \bigotimes_{i = 1}^d \wrapp{\bm{\lambda}^{(i)} \star \mu^{(\zeta)}}_i \right) \right) \right].
\end{equation}
Let $\mcM = \mcM(d, A)$ per~\cref{d:m-ad}.
By~\cref{l:expand-fb}, we have that 
$$\varphi\wrapp{F_B\wrapp{\mu^{(\zeta)}}} = \varphi \left(\sum_{\mfa \in \mcM} \zeta(\mfa) \cdot F_B\wrapp{\delta_{\mfa}^{\otimes d}} \right).$$
By~\cref{l:l1} and the observation that \[F_{B, \mfc^*}\wrapp{\mu^{(\zeta)}} = \frac{1}{1 + \max_{\mfb \in [q]} \sum_{\mfc \in [q]} A^d_{\mfc,\mfb}}, \]
the only way for~\cref{e:contradict} to hold for monotone $\varphi$ is if 
$$ 
\bigotimes_{i = 1}^d \wrapp{\bm{\lambda}^{(i)} \star \mu^{(\zeta)}}_i  = \delta_{\mfa}^{\otimes d} \text{ for some } \mfa \in \mcM
$$
for each collection $\wrapp{\bm{\lambda}^{(1)}, \ldots, \bm{\lambda}^{(d)}}$ in the support of $\xi$. In particular, since $\xi$ was assumed to be independent of $\zeta$, consider the following probability distribution induced over $\mcM$ by $\xi$:
\begin{align*}
    w_{\xi}(\mfa) \defeq \Pr_{\wrapp{\bm{\lambda}^{(1)},\dots,\bm{\lambda}^{(d)}} \sim \xi}\wrapb{\bigotimes_{i = 1}^d \wrapp{\bm{\lambda}^{(i)} \star \mu^{(\zeta)}}_i  = \delta_{\mfa}^{\otimes d}}, \qquad \forall \mfa \in \mcM.
\end{align*}
Then \cref{e:contradict} says that
\begin{align*}
    \varphi\wrapp{\sum_{\mfa \in \mcM} \zeta(\mfa) \cdot F_{B}\wrapp{\delta_{\mfa}^{\otimes d}}} = \sum_{\mfa \in \mcM} w_{\xi}(\mfa) \cdot \varphi\wrapp{F_{B}\wrapp{\delta_{\mfa}^{\otimes d}}}.
\end{align*}
On the one hand, the right-hand side is constant independent of $\zeta$ from our assumption that $\xi$ is universal. On the other hand, our second assumption on $A$ implies that $F_{B,\mfc}\wrapp{\delta_{\mfa_{1}}^{\otimes d}} \neq F_{B,\mfc}\wrapp{\delta_{\mfa_{2}}^{\otimes d}}$. Combined with monotonicity of $\varphi$, we see that the left-hand side is a nonconstant function of $\zeta$. Hence, we have a contradiction.

\end{proof}

\section{Evidence for Convexity of \texorpdfstring{$\mathcal{F}_{A,d}^{\PROD}$}{Fprod} for Antiferromagnetic Potts}
Our goal in this section is to prove \cref{thm:antiferro-potts-cvx-bulk}. We will construct the requisite product measure using the following technical lemma. Throughout this section, for a configuration $\tau \in [q]^{d}$, we write $\sgn(\tau) \in \Z_{\geq0}^{q}$ for the \emph{signature} of $\tau$, given by $\sgn(\tau)(\mfc) = \#\{i \in [d] : \tau(i) = \mfc\}$ for all $\mfc \in [q]$.
\begin{lemma}\label{lem:iid-Fmu-criterion}
Suppose $A = \bm{v}\bm{v}^{\top} - D$ for a positive vector $\bm{v} \in \R_{>0}^{q}$ and a diagonal matrix $D$ with entries $0 \leq D(\mfc,\mfc) \leq \bm{v}(\mfc)^{2}$. Fix a probability measure $\mu$ on $[q]^{d}$. Then the following are equivalent.
\begin{enumerate}
    \item[(1)] There exists an i.i.d. measure $\nu = \bm{p}^{\otimes d}$ for $\bm{p} \in \triangle_{[q]}$ such that $F(\mu) = F(\nu)$.
    \item[(2)] Let $\bm{y} \in \R_{\geq0}^{q}$ be given by $\bm{y}(\mfc) = \frac{D(\mfc,\mfc)}{\bm{v}(\mfc)^{2}} \in [0,1]$. Define the distribution $\xi$ on the slice $\wrapc{\alpha \in \Z_{\geq0}^{q} : \abs{\alpha} = d}$ by
    \begin{align*}
        \xi(\alpha) \propto \prod_{\mfb \in [q]} \bm{v}(\mfb)^{\alpha(\mfb)} \cdot \mu\wrapp{\wrapc{\tau \in [q]^{d} : \sgn(\tau) = \alpha}}.
    \end{align*}
    Then
    \begin{align}\label{eq:iid-Fmu-criterion}
        \sum_{\mfc \in [q]} \frac{1}{\bm{y}(\mfc)} \cdot \E_{\alpha \sim \xi}\wrapb{\wrapp{1 - \bm{y}(\mfc)}^{\alpha(\mfc)}}^{1/d} \geq \wrapp{-1 + \sum_{\mfc \in [q]} \frac{1}{\bm{y}(\mfc)}} \cdot \max_{\mfc \in [q]} \E_{\alpha \sim \xi}\wrapb{\wrapp{1 - \bm{y}(\mfc)}^{\alpha(\mfc)}}^{1/d}.
    \end{align}
\end{enumerate}
\end{lemma}
It is known that \cref{eq:iid-Fmu-criterion} does not hold for arbitrary probability distributions $\mu$ over $[q]^{d}$, even for the antiferromagnetic Potts model in the correlation decay regime. For example, one can take $\mu = \delta_{\tau}$ where $\tau \in [q]^{d}$ is chosen so that for some color $\mfc \in [q]$, $\alpha(\mfc) = \#\{i \in [d] : \tau(i) = \mfc\} = 0$. Hence, to certify that $\mathcal{F}_{A,d}^{\PROD} = \mathcal{F}_{A,d}$ and that $\mathcal{F}_{A,d}^{\PROD}$ is convex, one needs to look beyond i.i.d. measures $\nu$. However, we show that even with weak lower tail bounds on the random variables $\alpha(\mfc)$ for $\mfc \in [q]$, \cref{eq:iid-Fmu-criterion} holds. We now use \cref{lem:iid-Fmu-criterion} to prove \cref{thm:antiferro-potts-cvx-bulk}; a proof of \cref{lem:iid-Fmu-criterion} is provided at the end of this section.
\begin{proof}[Proof of \cref{thm:antiferro-potts-cvx-bulk}]
To construct our desired product measure $\nu$, we will verify (2) of \cref{lem:iid-Fmu-criterion} and apply the lemma to obtain $\nu$. In the setting of \cref{lem:iid-Fmu-criterion}, we take $v = \allone$, $D = (1 - \beta)I$ and $y = (1 - \beta) \allone$, in which case \cref{eq:iid-Fmu-criterion} simplifies drastically to
\begin{align}\label{eq:antiferro-potts-iid-criterion}
    \sum_{\mfc \in [q]} \E_{\alpha \sim \xi}\wrapb{\beta^{\alpha(\mfc)}}^{1/d} \geq (q - (1 - \beta)) \cdot \max_{\mfc \in [q]} \E_{\alpha \sim \xi}\wrapb{\beta^{\alpha(\mfc)}}^{1/d}
\end{align}
where $\xi(\alpha) = \mu\wrapp{\wrapc{\tau \in [q]^{d} : \sgn(\tau) = \alpha}}$ is the distribution over signatures induced by $\mu$ via marginalization. We will establish \cref{eq:antiferro-potts-iid-criterion} using our assumed tail bounds \cref{eq:color-density-lower-tail}, as well as the fact that we are in the uniqueness regime.

Let $\mfc^{*}$ attain the maximum in the right-hand side of \cref{eq:antiferro-potts-iid-criterion}. Rearranging, the desired inequality is equivalent to
\begin{align*}
    \frac{1}{q - 1} \sum_{\mfc \neq \mfc^{*}} \frac{\E_{\alpha \sim \xi}\wrapb{\beta^{\alpha(\mfc)}}^{1/d}}{\E_{\alpha \sim \xi}\wrapb{\beta^{\alpha(\mfc^{*})}}^{1/d}} \geq 1 - \frac{1 - \beta}{q - 1}.
\end{align*}
To verify this inequality, observe that
\begin{align*}
    \frac{1}{q - 1} \sum_{\mfc \neq \mfc^{*}} \E_{\alpha \sim \xi}\wrapb{\beta^{\alpha(\mfc)}}^{1/d} &\geq \frac{1}{q - 1} \sum_{\mfc \neq \mfc^{*}} \E_{\alpha \sim \xi}\wrapb{\beta^{\alpha(\mfc)/d}} \tag{Jensen $+$ concavity of $x \mapsto x^{1/d}$} \\
    &\geq \E_{\alpha \sim \xi}\wrapb{\beta^{\frac{1}{q-1}\sum_{\mfc \neq \mfc^{*}} \frac{\alpha(\mfc)}{d}}} \tag{Jensen + convexity of $x \mapsto \beta^{x}$} \\
    &= \beta^{\frac{1}{q-1}} \cdot \E_{\alpha \sim \xi}\wrapb{\beta^{-\frac{1}{q-1} \cdot \frac{\alpha(\mfc^{*})}{d}}}. \tag{Using $\sum_{\mfc \in [q]} \alpha(\mfc) = d$}
\end{align*}
Hence, it suffices to verify the inequality
\begin{align*}
    \frac{\E_{\alpha \sim \xi}\wrapb{\beta^{-\frac{1}{q-1} \cdot \frac{\alpha(\mfc^{*})}{d}}}}{\E_{\alpha \sim \xi}\wrapb{\beta^{\alpha(\mfc^{*})}}^{1/d}} \geq \beta^{-\frac{1}{q-1}}\wrapp{1 - \frac{1-\beta}{q-1}}.
\end{align*}
Plugging in our assumptions, we have that the left-hand side above is lower bounded as
\begin{align*}
    \frac{\epsilon + (1 - \epsilon) \beta^{-\frac{1}{q-1} \cdot \frac{\gamma_{*}}{q}}}{\wrapp{\epsilon + (1 - \epsilon)\beta^{\gamma_{*} \cdot \frac{d}{q}}}^{1/d}} &\geq \beta^{-\frac{\gamma_{*}}{q-1}} \cdot \frac{\epsilon \beta^{\frac{1}{q-1} \cdot \frac{\gamma_{*}}{q}} + (1 - \epsilon)}{\wrapp{\epsilon \beta^{-\gamma_{*} \cdot \frac{d}{q}} + (1 - \epsilon)}^{1/d}} \\
    &\geq \beta^{-\frac{\gamma_{*}}{q-1}} \cdot \frac{\epsilon \beta^{\frac{1}{10q(q-1)}} + (1 - \epsilon)}{\wrapp{\epsilon \beta^{-\frac{d}{10q}} + (1 - \epsilon)}^{1/d}}. \tag{$\gamma_{*} \leq \frac{1}{10}$ and monotonicity in $\gamma_{*}$}
\end{align*}
Hence, we require that $\gamma_{*},\epsilon,\beta$ satisfy
\begin{align*}
    \gamma_{*} \geq 1 - \frac{q-1}{\log \frac{1}{\beta}} \log\frac{q-1}{(q - 1) - (1 - \beta)} + \underset{\defeq \Xi_{d,q,\epsilon}(\beta)}{\underbrace{\frac{q-1}{\log \frac{1}{\beta}}\log \frac{\wrapp{\epsilon \beta^{-\frac{d}{10q}} + (1 - \epsilon)}^{1/d}}{\epsilon \beta^{\frac{1}{10q(q-1)}} + (1 - \epsilon)}}}
\end{align*}
This inequality follows immediately by combining the following numerical inequalities, whose proofs are deferred to \cref{sec:numerical}.
\begin{claim}\label{claim:weird}
For every $q \geq 2$ and every $0 \leq \beta \leq 1$, we have
\begin{align*}
    \frac{q-1}{\log \frac{1}{\beta}} \log\frac{q-1}{(q - 1) - (1 - \beta)} \geq \sqrt{\beta}.
\end{align*}
\end{claim}
\begin{claim}\label{claim:insane-function}
Suppose $d \geq 2q$ and $\epsilon \leq 1/4$. Then $\Xi_{d,q,\epsilon}(\beta)$ is a decreasing function in $\beta$ on the interval $\wrapb{1 - \frac{q}{d + 1}, 1}$, and
\begin{align*}
    \Xi_{d,q,\epsilon}\wrapp{1 - \frac{q}{d+1}} \leq \epsilon.
\end{align*}
\end{claim}
\end{proof}

\begin{proof}[Proof of \cref{lem:iid-Fmu-criterion}]
Since $F(\mu) \propto G(\mu)$, $F(\nu) \propto G(\nu)$, and the unnormalized recursion $G(\cdot)$ is homogeneous, (1) holds if and only if there exists a nonzero nonnegative vector $\bm{x} \in \R_{\geq0}^{q}$ such that $G\wrapp{\bm{x}^{\otimes d}} \propto G(\mu)$, i.e.
\begin{align*}
    \langle \bm{x}, A(\mfc,\cdot) \rangle^{d} = s \cdot G_{\mfc}(\mu), \qquad \forall \mfc \in [q],
\end{align*}
for some common scalar $s > 0$ independent of $\mfc$. Indeed, if such a $\bm{x}$ exists, then we can take $\bm{p} = \bm{x} / \norm{\bm{x}}_{1}$ and $\nu = \bm{p}^{\otimes d}$ to obtain $F(\mu) = F(\nu)$.

By rearranging, the above is equivalent to the claim that
\begin{align*}
    A^{-1}G(\mu)^{1/d} \in \R_{\geq0}^{q},
\end{align*}
where as usual, $G(\mu)^{1/d}$ means we raise each entry of $G(\mu)$ to the $1/d$ power. Now observe that if $\zeta$ denotes the distribution over $\{\alpha \in \Z_{\geq0}^{q} : \abs{\alpha} = d\}$ given by $\zeta(\alpha) = \mu\wrapp{\wrapc{\tau \in [q]^{d} : \sgn(\tau) = \alpha}}$, then
\begin{align*}
    G_{\mfc}(\mu) &= \E_{\tau \sim \mu}\wrapb{\prod_{i=1}^{d} A_{\mfc,\tau(i)}} \tag{Definition} \\
    &= \E_{\alpha \sim \zeta}\wrapb{\prod_{\mfb \in [q]} A_{\mfc,\mfb}^{\alpha(\mfb)}} \\
    &= \E_{\alpha \sim \zeta}\wrapb{\prod_{\mfb \neq \mfc} \bm{v}(\mfb)^{\alpha(\mfb)}\bm{v}(\mfc)^{\alpha(\mfb)} \cdot \wrapp{\bm{v}(\mfc)^{2} - D(\mfc,\mfc)}^{\alpha(\mfc)}} \tag{Assumption $A = \bm{v}\bm{v}^{\top} - D$} \\
    &= \bm{v}(\mfc)^{d} \cdot \E_{\alpha \sim \zeta}\wrapb{\prod_{\mfb \in [q]} \bm{v}(\mfb)^{\alpha(\mfb)} \cdot \wrapp{1 - \frac{D(\mfc,\mfc)}{\bm{v}(\mfc)^{2}}}^{\alpha(\mfc)}} \\
    &\propto \bm{v}(\mfc)^{d} \cdot \E_{\alpha \sim \xi}\wrapb{\wrapp{1 - \bm{y}(\mfc)}^{\alpha(\mfc)}}. \tag{Definition of $\xi$ and $\bm{y}$}
\end{align*}
In particular,
\begin{align*}
    A^{-1}G(\mu)^{1/d} \propto A^{-1}\diag(\bm{v}) \wrapb{\E_{\alpha \sim \xi}\wrapb{\wrapp{1 - \bm{y}(\mfc)}^{\alpha(\mfc)}}^{1/d}}_{\mfc \in [q]}.
\end{align*}
It follows that (1) is equivalent to the right-hand side lying in $\R_{\geq0}^{q}$. Now observe that we can write
\begin{align*}
    A^{-1}\diag(\bm{v}) &= \diag(\bm{v})^{-1} \diag(\bm{v}) \wrapp{\bm{v}\bm{v}^{\top} - D}^{-1} \diag(\bm{v}) \\
    &= \diag(\bm{v})^{-1} \wrapp{\allone\allone^{\top} - \diag(\bm{y})}^{-1} \\
    &= \diag(\bm{v})^{-1} \wrapp{-\diag(1/\bm{y}) - \frac{\diag(1/\bm{y})\allone\allone^{\top}\diag(1/\bm{y})}{1 - \sum_{\mfc \in [q]} \bm{y}(\mfc)^{-1}}}. \tag{Sherman--Morrison Formula}
\end{align*}
Hence, (1) is equivalent to
\begin{align*}
    \wrapp{-\diag(1/\bm{y}) - \frac{\diag(1/\bm{y})\allone\allone^{\top}\diag(1/\bm{y})}{1 - \sum_{\mfc \in [q]} \bm{y}(\mfc)^{-1}}} \cdot \wrapb{\E_{\alpha \sim \xi}\wrapb{\wrapp{1 - \bm{y}(\mfc)}^{\alpha(\mfc)}}^{1/d}}_{\mfc \in [q]} \geq 0
\end{align*}
entrywise. Rearranging precisely yields (2).
\end{proof}

\section{Future Directions}
In view of \cref{t:main}, it's natural to consider possible alternatives to a Weitz-style reduction that might still hold. An essential part of what makes Weitz's argument work in the 2-spin case is the essentially one-dimensional nature of the image of $F$. Indeed, in the $2$-spin case, $F$ outputs a distribution on $\{0,1\}$, which can be parameterized by a single number. In the language of \cref{defn:calF-and-calFprod}, since $\cF_{A,d}^{\PROD}$ and $\cF_{A,d} = \conv\wrapp{\cF_{A,d}^{\PROD}}$ can both be viewed as subsets of $[0,1]$, a trivial continuity argument shows that $\cF_{A,d} = \cF_{A,d}^{\PROD}.$ As we have shown in this paper, this argument necessarily breaks down in the multispin setting. However, in existing analyses, the \emph{reason} one wants to reduce to the tree is to perform a contraction analysis there with respect to a given (pseudo)metric. Hence, at the end of the day, we are primarily interested in a one-dimensional functional of the image of $F,$ some distance function. This motivates the following question.
\begin{question}
    Is it the case that for every pair of positive integers $q\geq 3,d\geq 2$ and every $q$-spin interaction matrix $A \in \R^{q\times q}_{\geq 0},$ there exists a bounded potential (see \cref{defn:bdd-potential}) $\varphi:[0,1]\to\R$, a universal distribution $\xi$ on $d$-tuples of external fields $\wrapp{\bm{\lambda}^{(1)}, \dots, \bm{\lambda}^{(d)}}$, where $\bm{\lambda}^{(i)}\in \R_{\geq0}^{d\times q}$ for each $i\in [d]$, and a (pseudo)metric $\mathscr{D}$ on $\R^q$ such that for all probability measures $\mu$ and $\nu$ on $[q]^d$,
    \begin{align*}
        \mathscr{D}(\varphi(F_A(\mu)), &\varphi(F_A(\nu))) \\
        &\leq \E_{(\bm{\lambda}^{(1)}, \dots, \bm{\lambda}^{(d)})\sim \xi}\left[\mathscr{D}\left( \varphi \left( F_A \left( \bigotimes_{i = 1}^d \wrapp{\bm{\lambda}^{(i)} \star \mu}_i \right) \right) ,  \varphi \left( F_A \left( \bigotimes_{i = 1}^d \wrapp{\bm{\lambda}^{(i)} \star \nu}_i \right) \right) \right)\right]\text{?}
    \end{align*}
\end{question}
An affirmative answer to this question would let one convert contraction with respect to $\mathscr{D}(\cdot,\cdot)$ under potential $\varphi$ on trees to strong spatial mixing for general graphs via a similar argument to that in~\cref{app:main-q-consequences}.

A related interesting direction would be to more closely link the agreement of $\cF_{A,d}^{\PROD}$ and  $\cF_{A,d}$ to the existence of universal fields as in~\cref{prop:weitz-reinterpret}.
\begin{question}\label{q:families}
    Is it the case that, for any $q$-spin system interaction matrix $A$ for which $\cF_{A,d}^{\PROD}= \cF_{A,d}$, there exist universal fields as in \cref{prop:weitz-reinterpret}?
\end{question}

The above question is of particular interest due to~\cref{thm:antiferro-potts-cvx-bulk}, which gives some evidence for $\mathcal{F}_{A,d}^{\PROD} = \mathcal{F}_{A,d}$ for the antiferromagnetic Potts model in large maximum degree graphs. This motivates the following natural question.

\begin{question}\label{q:potts}
Fix an integer $q \ge 3$ and consider the antiferromagnetic Potts interaction matrix $A = \ind_q \ind_q^{\top} - (1 - \beta)I_q$ with $\max \{0, 1 - \frac{q}{d + 1}\} \le \beta \le 1$. Then, is $\cF_{A, d}^{\PROD}$ a convex set?
\end{question}

An affirmative answer to both~\cref{q:potts} and~\cref{q:families} would resolve~\cref{c:jerrum} as formalized in~\cref{app:main-q-consequences}.

\printbibliography

\appendix
\section{Belief propagation on general graphs}\label{app:BP-general}

%\section{Weitz's SAW tree as a random computation tree}\label{app:weitz-reinterpretation}
\begin{proof}[Proof of \cref{prop:weitz-reinterpret}]
Write $\{0,1\}$ for the set of spins. The key idea behind the Weitz trick is to expand the marginal \emph{ratio}, crucially taking advantage of the fact that $q = 2$, so any marginal distribution has the form $[p,1-p]$ for some $p \in [0,1]$. We have
\begin{align*}
    \frac{F_{0}(\mu)}{F_{1}(\mu)} &= \frac{G_{0}(\mu)}{G_{1}(\mu)} \\
    &= \frac{\E_{\tau \sim \mu}\wrapb{\prod_{i=1}^{d} A_{0,\tau(i)}}}{\E_{\tau \sim \mu}\wrapb{\prod_{i=1}^{d} A_{1,\tau(i)}}} \\
    &= \prod_{i=1}^{d} \frac{\E_{\tau \sim \mu}\wrapb{\prod_{j=1}^{i} A_{0,\tau(j)} \prod_{j=i+1}^{d} A_{1,\tau(j)}}}{\E_{\tau \sim \mu}\wrapb{\prod_{j=1}^{i-1} A_{0,\tau(j)} \prod_{j=i}^{d} A_{1,\tau(j)}}} \tag{Telescoping} \\
    &= \frac{\prod_{i=1}^{d} \wrapp{\tilde{\mu}_{i}(0) \cdot A_{0,0} + \tilde{\mu}_{i}(1) \cdot A_{0,1}}}{\prod_{i=1}^{d} \wrapp{\tilde{\mu}_{i}(0) \cdot A_{1,0} + \tilde{\mu}_{i}(1) \cdot A_{1,1}}} \tag{$\ast$} \\
    &= \frac{F_{0}\wrapp{\bigotimes_{i=1}^{d} \tilde{\mu}_{i}}}{F_{1}\wrapp{\bigotimes_{i=1}^{d} \tilde{\mu}_{i}}}. \tag{$\ast\ast$}
\end{align*}
Here, $(\ast)$ is justified by dividing the numerator and denominator of the $i$th term in the product by
\begin{align*}
    \E_{\tau \sim \mu}\wrapb{\prod_{j=1}^{i-1} A_{0,\tau(j)} \prod_{j=i+1}^{d} A_{1,\tau(j)}}.
\end{align*}
In particular, the distribution $\tilde{\mu}_{i}$ on $\{0,1\}$ is defined by beginning with the distribution on $\{0, 1\}^{[d]}$
\begin{align*}
    \wrapp{\bm{\lambda}^{(i)} \star \mu} \propto \mu(\tau) \prod_{j=1}^{i-1} A_{0,\tau(j)} \prod_{j=i+1}^{d} A_{1,\tau(j)}
\end{align*}
on $\{0,1\}^{d}$ and marginalizing out $[d] \setminus \{i\}$. We express this succinctly as $\tilde{\mu}_{i} = \wrapp{\bm{\lambda}^{(i)} \star \mu}_{i}$ for all $i=1,\dots,d$. Here, the external field $\bm{\lambda}^{(i)}$ is given by $\bm{\lambda}_{j,0}^{(i)} = A_{0,0}, \bm{\lambda}_{j,1}^{(i)} = A_{0,1}$ for $j < i$, $\bm{\lambda}^{(i)}_i = \ind$, and $\bm{\lambda}_{j,0}^{(i)} = A_{1,0}, \bm{\lambda}_{j,1}^{(i)} = A_{1,1}$ for $j>i.$\footnote{One can interpret $\tilde{\mu}_{i}$ combinatorially as follows. Suppose we are trying to compute the marginal distribution of a vertex $r$ in a graph. The new marginal distributions $\tilde{\mu}_{i}$ are obtained combinatorially by splitting the vertex $r$ into $d$ copies $r_{1},\dots,r_{d}$, attaching $r_{i}$ to $u_{i}$ as a leaf node, pinning the prefix $r_{1},\dots,r_{i-1} \gets 0$, and pinning the suffix $r_{i+1},\dots,r_{d} \gets 1$. This description is more in line with the original presentation of Weitz's self-avoiding walk tree.} We stress that $\bm{\lambda}^{(i)}$ does not depend on the input distribution $\mu$.

We obtained $(\ast\ast)$ by dividing the numerator and denominator by their sum. It immediately follows that for the product measure $\bigotimes_{i=1}^{d} \tilde{\mu}_{i} = \bigotimes_{i=1}^{d} \wrapp{\bm{\lambda}^{(i)} \star \mu}_{i}$, we have
\begin{align*}
    F(\mu) = F\wrapp{\bigotimes_{i=1}^{d} \wrapp{\bm{\lambda}^{(i)} \star \mu}_{i}}
\end{align*}
as desired.
\end{proof}

\begin{proof}[Proof of \cref{lem:recursion-general}]
At the level of partition functions, writing $Z_{G}^{S \gets \tau}$ for the contribution of all terms $\sigma : V \to [q]$ satisfying $\sigma\mid_{S} = \tau$ for $\tau : S \to [q]$, we have
\begin{align*}
    Z_{G}^{r \gets \mfc} = \bm{\lambda}_{r,\mfc} \sum_{\tau : N(r) \to [q]} \prod_{i=1}^{d} A_{\mfc,\tau(u_{i})} \cdot Z_{G-r}^{N(r) \gets \tau}.
\end{align*}
It follows that
\begin{align*}
    \mu_{G,r}(\mfc) &= \frac{Z_{G}^{r \gets \mfc}}{Z_{G}} = \frac{\bm{\lambda}_{r,\mfc} \sum_{\tau : N(r) \to [q]} \prod_{i=1}^{d} A_{\mfc,\tau(u_{i})} \cdot Z_{G-r}^{N(r) \gets \tau}}{\sum_{\mfa \in [q]} \bm{\lambda}_{r,\mfa} \sum_{\tau : N(r) \to [q]} \prod_{i=1}^{d} A_{\mfa,\tau(u_{i})} \cdot Z_{G-r}^{N(r) \gets \tau}} \\\
    &= \frac{G_{\mfc}\wrapp{\mu_{G-r,N(r)}}}{\sum_{\mfa \in [q]} G_{\mfa}\wrapp{\mu_{G-r,N(r)}}}.
\end{align*}
Here, the final equality follows by dividing by $Z_{G-r}$ in both the numerator and denominator, once for each neighbor $u_{1},\dots,u_{d}$. The claim follows.
\end{proof}

%Recall from~\cref{eq:F-arbitrary} and~\cref{eq:F-prod} that given interaction matrix $A$ with $q \ge 3$ spins and some positive integer $d$, if $\Delta_q$ denote the space of probability distributions over $[q]$ we have 
%\begin{align*}
%    \cF_{A,d}^{\PROD} \defeq \left\{ F_A(\nu) : \nu \text{ is a product measure over } [q]^d \right\} \subseteq \Delta_q,  \\
%    \cF_{A,d} \defeq \left\{ F_A(\mu) : \mu \text{ is an arbitrary measure over } [q]^d \right\} \subseteq \Delta_q. 
%\end{align*}
%We verify that $\conv(\cF^{\PROD}_{A, d} ) = \cF_{A, d}$.
Finally, recall the sets $\cF_{A,d}^{\PROD},\cF_{A,d}$ from \cref{eq:F-prod,eq:F-arbitrary}. We now show that these two sets are intimately related to one another.
\begin{lemma}\label{lem:image-F-cvxhull}
For every $q \geq 2$ and every interaction matrix $A \in \R_{\geq0}^{q \times q}$, $\cF_{A,d}$ is a convex polytope and given by
\begin{align*}
     \cF_{A,d} = \conv\wrapp{\cF_{A,d}^{\PROD}} = \conv\wrapc{F_{A,d}(\delta_{\tau}) : \tau \in [q]^{d}}.
\end{align*}
\end{lemma}
\begin{proof}
If we extend the domain of $F$ and $G$ to all of $\R_{\geq0}^{[q]^{d}}$ in the straightforward way, then it is clear that
\begin{align*}
    \cF_{A,d} = \wrapc{G\wrapp{\bm{x}} : \bm{x} \in \R_{\geq0}^{[q]^{d}}} \cap \triangle_{q}.
\end{align*}
Since $G$ is a linear function in $\bm{x}$, the first set is a convex polyhedral cone contained in $\R_{\geq0}^{q}$. Intersecting with $\triangle_{q}$ yields a convex polytope, establishing the first claim. For the second claim, observe that for any distribution $\mu$ over $[q]^{d}$,
\begin{align*}
    F(\mu) = \sum_{\tau \in [q]^{d}} \xi_{\mu}(\tau) \cdot F\wrapp{\delta_{\tau}}
\end{align*}
where
\begin{align*}
    \xi_{\mu}(\tau) \propto \mu(\tau) \sum_{\mfc \in [q]} \prod_{i=1}^{d} A_{\mfc,\tau(i)}, \qquad \forall \tau \in [q]^{d}.
\end{align*}
This completes the proof.
\end{proof}

\section{Consequences of a Positive Answer to the Main Question}\label{app:main-q-consequences}
In this section, we show that if \cref{q:main} had a positive answer, then contraction for belief propagation on product measures, i.e. correlation decay on trees, implies a strong form of correlation decay for general graphs, as well as optimal mixing of Glauber dynamics via spectral independence. This motivates \cref{q:main} as a unified approach to tackling \cref{metaQ:main}. Let us first formalize the notion of contraction commonly used. As in prior works (see e.g. \cite{LLY13, RSTVY13, CLV20, CLMM23}), we will use the potential method with respect to a suitably controlled potential function.

\begin{definition}[Bounded Potential]\label{defn:bdd-potential}
We say a function $\varphi : [0,1] \to \R_{\geq0}$ is a \emph{bounded potential} if it is smooth, strictly monotone, and if there exist positive universal constants $L,L' > 0$ such that $\abs{\varphi(p)} \leq L, \abs{\varphi'(p)} \geq L'$ for all $p \in [0,1]$. Throughout, we write $\Phi = \varphi'$. If $\bm{p} \in [0,1]^{q}$ is a vector, then we write $\varphi(\bm{p})$ to mean entrywise application of the potential. We apply the same convention to its inverse $\varphi^{-1}$.
\end{definition}
\begin{remark}
It will be clear in a moment why these boundedness assumptions are necessary. In practice, they often only need to be enforced on a smaller interval in $[0,1]$, which can derived from standard \emph{marginal bounds} for the spin system under consideration. We assume boundedness throughout $[0,1]$ for simplicity.
\end{remark}
\begin{definition}[Contraction]\label{def:contraction}
Fix a symmetric interaction matrix $A \in \R_{\geq0}^{q \times q}$. For $\Delta \in \N$, we say $F_{A}$ is \emph{up-to-$\Delta$ contracting} if there exists a norm $\norm{\cdot}_{K}$ on $\R^{q}$, a bounded potential $\varphi : [0,1] \to \R_{\geq0}$, and $0 < \delta < 1$ such that for every integer $1 \leq d < \Delta$, the function
\begin{align*}
    F_{A,d}^{\circ \varphi}\wrapp{\bm{m}_{1},\dots,\bm{m}_{d}} \defeq \varphi\wrapp{F_{A,d}\wrapp{\bigotimes_{i=1}^{d} \varphi^{-1}(\bm{m}_{i})}}%, \qquad \forall \bm{m}_{1},\dots,\bm{m}_{d} \in \varphi\wrapp{\triangle_{q}}
\end{align*}
has Jacobian $J$ (a matrix in $\R^{q \times dq}$) satisfying the following: For every $\bm{m}_{1},\dots,\bm{m}_{d} \in \varphi\wrapp{\triangle_{q}}$, its induced matrix norm is bounded as
\begin{align}\label{eq:jacobian-contraction}
    \norm{JF_{A,d}^{\circ \varphi}(\bm{m}_{1},\dots,\bm{m}_{d})}_{\norm{\cdot}_{K} \to \norm{\cdot}_{\infty,K}} \leq 1 - \delta,
\end{align}
where we endow the range $\R^{q}$ of matrix $JF_{A,d}^{\circ \varphi}(\bm{m}_{1},\dots,\bm{m}_{d})$ with the given norm $\norm{\cdot}_{K}$, and the domain $\R^{dq} \cong \wrapp{\R^{q}}^{d}$ with the \emph{hybrid} vector norm $\norm{\bm{x}}_{\infty,K} \defeq \max_{i=1,\dots,d} \norm{\bm{x}_{i}}_{K}$.
\end{definition}
\begin{remark}
By the multivariate Mean Value Theorem, the above Jacobian condition readily implies
\begin{align}\label{eq:contraction}
    \norm{\varphi\wrapp{F_{A}(\mu)} - \varphi\wrapp{F_{A}(\nu)}}_{K} \leq (1 - \delta) \max_{i=1,\dots,d} \norm{\varphi\wrapp{\mu_{i}} - \varphi\wrapp{\nu_{i}}}_{K}
\end{align}
for every $1 \leq d < \Delta$ and pair of product measures $\mu = \bigotimes_{i=1}^{d} \mu_{i}, \nu = \bigotimes_{i=1}^{d} \nu_{i}$. All the analysis in this section can be made robust enough to handle weighted variants of $\norm{\cdot}_{\infty,K}$ like those used in \cite{CLMM23}.
\end{remark}

\subsection{Contraction Implies Correlation Decay for General Graphs}\label{sec:main-conj-implies-SSM}
We begin by showing how contraction on trees plus a positive answer to \cref{q:main} together imply a very strong form of correlation decay known as \emph{strong spatial mixing} for \emph{general} bounded-degree graphs.
\begin{defn}[Strong Spatial Mixing]\label{defn:ssm}
We say the Gibbs distribution $\mu$ of a $q$-spin system on a graph $G=(V,E)$ exhibits \emph{strong spatial mixing (SSM)} if there exists $0 < \delta < 1$ and $C > 0$ such that for every vertex $r \in V$, every subset $\Lambda \subseteq V \setminus\{r\}$, and every pair of external fields $\bm{\lambda},\bm{\lambda}' \in \R_{\geq0}^{\Lambda \times [q]}$ which differ on $\partial(\bm{\lambda},\bm{\lambda}') \defeq \{v \in \Lambda : \bm{\lambda}_{v,\cdot} \neq \bm{\lambda}'_{v,\cdot}\} \subseteq \Lambda$, we have the inequality
\begin{align*}
    \norm{\wrapp{\bm{\lambda} \star \mu}_{r} - \wrapp{\bm{\lambda}' \star \mu}_{r}}_{\TV} \leq C \cdot (1-\delta)^{\dist_{G}(r,\partial(\bm{\lambda},\bm{\lambda}'))}.
\end{align*}
where $\|\cdot\|_{\TV}$ denotes the total variation distance, and $\dist_{G}(\cdot,\cdot)$ denotes the shortest path distance in $G$.
\end{defn}
\begin{remark}
Technically, the above definition is a strengthening of the usual notion of strong spatial mixing which only considers the case where $\bm{\lambda},\bm{\lambda}'$ correspond to \emph{pinnings}, i.e. there exist $\tau, \tau' : \Lambda \to [q]$ such that for every $v \in \Lambda$, $\bm{\lambda}_{v,\tau(v)} = \bm{\lambda}'_{v,\tau'(v)} = \infty$ and $\bm{\lambda}_{v,\mfc} = 0$ for all $\mfc \neq \tau(v)$ (resp. $\bm{\lambda}'_{v,\mfc} = 0$ for all $\mfc \neq \tau'(v)$).
\end{remark}

\begin{theorem}
Fix a symmetric interaction matrix $A \in \R_{\geq0}^{q \times q}$. Suppose $A$ is up-to-$\Delta$ contracting with respect to a norm $\norm{\cdot}_{K}$ on $\R^{q}$ and a bounded potential $\varphi$, and further assume that \cref{q:main} holds for $A$ with respect to the potential $\varphi$. Then the Gibbs distribution of the spin system associated to $A$ on any graph of maximum degree $\Delta$ exhibits strong spatial mixing.
\end{theorem}
\begin{remark}
It is known that in general, enforcing ``up-to-$\Delta$ contraction'' as opposed to only contraction for $d = \Delta - 1$ is necessary for strong spatial mixing on general graphs of maximum degree $\Delta$; see \cite{LLY13}.
\end{remark}
\begin{proof}
Observe that by combining contraction, convexity of the norm $\norm{\cdot}_{K}$, and the fact that the averaging distribution $\xi$ does not depend on the choice of input distribution in \cref{q:main}, we have that for any $d \in \N$ and any pair of measures $\mu,\nu$ on $[q]^{d}$
\begin{align}\label{eq:contract-to-iterate}
    \norm{\varphi\wrapp{F_{A}(\mu)} - \varphi\wrapp{F_{A}(\nu)}}_{K} \le (1 - \delta) \max_{i=1,\dots,d} \sup_{\bm{\lambda} \in \R_{\geq0}^{d \times q}} \norm{\varphi\wrapp{\wrapp{\bm{\lambda} \star \mu}_{i}} - \varphi\wrapp{\wrapp{\bm{\lambda} \star \nu}_{i}}}_{K}.
\end{align}
Now fix a vertex $r \in V$ with neighbors $N(r) = \{u_{1},\dots,u_{d}\}$ and fix generalized boundary conditions $\bm{\lambda},\bm{\lambda}' \in \R_{\geq0}^{\Lambda \times [q]}$ for $\Lambda \subseteq V \setminus \{r\}$. We have % Writing $K = \dist\wrapp{r,\partial(\bm{\lambda},\bm{\lambda}')}$, we have
\begin{align*}
    &\norm{\wrapp{\bm{\lambda} \star \mu_{G}}_{r} - \wrapp{\bm{\lambda}' \star \mu_{G}}_{r}}_{\TV} \\
    &\leq C_{\varphi,K} \cdot \norm{\varphi\wrapp{\wrapp{\bm{\lambda} \star \mu_{G}}_{r}} - \varphi\wrapp{\wrapp{\bm{\lambda}' \star \mu_{G}}_{r}}}_{K} \tag{Boundedness of $\varphi$} \\
    &\leq C_{\varphi,K}(1 - \delta) \max_{i=1,\dots,d} \sup_{\bm{\lambda}^{(i)} \in \R_{\geq0}^{N(r) \times [q]}} \norm{\varphi\wrapp{\bm{\lambda}^{(i)} \star \wrapp{\bm{\lambda} \star \mu_{G-r}}_{u_{i}}} - \varphi\wrapp{\bm{\lambda}^{(i)} \star \wrapp{\bm{\lambda}' \star \mu_{G-r}}_{u_{i}}}}_{K} \tag{Using \cref{eq:contract-to-iterate} with the measures $(\bm{\lambda} \star \mu_{G})_{r}, (\bm{\lambda}' \star \mu_{G})_{r}$ over $[q]^{d}$} \\
    &\leq C_{\varphi,K}(1 - \delta) \max_{i=1,\dots,d} \sup_{\bm{\lambda}^{(i)} \in \R_{\geq0}^{N(r) \times [q]}} \norm{\varphi\wrapp{\wrapp{\bm{\lambda}^{(i)} \star \bm{\lambda} \star \mu_{G-r}}_{u_{i}}} - \varphi\wrapp{\wrapp{\bm{\lambda}^{(i)} \star \bm{\lambda}' \star \mu_{G-r}}_{u_{i}}}}_{K} \tag{Merging external fields} \\
    &\leq \dotsb \tag{Induction} \\
    &\leq C_{\varphi,K}' (1 - \delta)^{\dist(r,\partial(\bm{\lambda},\bm{\lambda}'))}. \tag{Invoking boundedness again}
%    &\leq C_{\varphi,p}(1 - \delta)^{\dist(r,\partial(\bm{\lambda},\bm{\lambda}'))} \cdot \max_{v \in \partial(\bm{\lambda},\bm{\lambda}')} \max_{S \subseteq B(r,K)} \sup_{\bm{\lambda}^{(S)} \in \R_{\geq0}^{S \times [q]}} \norm{\varphi\wrapp{\wrapp{\bm{\lambda}^{(S)} \star \bm{\lambda} \star \mu_{G[V\setminus S]}}_{v}} - \varphi\wrapp{\wrapp{\bm{\lambda}^{(S)} \star \bm{\lambda}' \star \mu_{G[V\setminus S]}}_{v}}}_{K} \\
\end{align*}
Crucially, throughout the induction, we used the fact that the new external fields we introduce through \cref{eq:contract-to-iterate} do not change the set of vertices $\partial(\bm{\lambda},\bm{\lambda}')$ with disagreeing fields. For instance, we have $\partial\wrapp{\bm{\lambda}^{(i)} \sqcup \bm{\lambda}, \bm{\lambda}^{(i)} \sqcup \bm{\lambda}'} = \partial(\bm{\lambda},\bm{\lambda}')$ for the fourth line, where $\sqcup$ denotes concatenation.
\end{proof}

\subsection{Contraction Implies Spectral Independence and Optimal Mixing}\label{sec:main-conj-implies-SI}
In this subsection, we establish that contraction on trees plus a positive answer to \cref{q:main} implies optimal mixing of Glauber dynamics via the spectral independence framework introduced in \cite{ALO21}. We work with the notation from \cite{CLMM23}.
\begin{defn}[Influence Matrices]
Let $\mu$ be the Gibbs distribution of a $q$-spin system on a finite graph $G=(V,E)$. Let $\Lambda \subseteq V$ be a subset of vertices, and let $\tau : \Lambda \to [q]$ be a valid pinning. For vertices $r,v \in V \setminus \Lambda$ and colors $\mfb,\mfc \in [q]$, we define the \emph{influence} of the vertex-color pair $(r,\mfb)$ on the vertex-color pair $(v,\mfc)$ by the quantity
\begin{align*}
    \infl_{\mu^{\tau}}((r,\mfb) \to (v,\mfc)) &\defeq \mu_{v}^{\tau, r \gets \mfb}(\mfc) - \mu_{v}^{\tau}(\mfc) \\
    &= \Pr_{\sigma \sim \mu}[\sigma(v) = \mfc \mid \sigma_{\Lambda} = \tau, \sigma(r) = \mfb] - \Pr_{\sigma \sim \mu}[\sigma(v) = \mfc \mid \sigma_{\Lambda} = \tau].
\end{align*}
If $r=v$, or if the assignments $r \gets \mfb$ or $v \gets \mfc$ are not feasible in $\mu^{\tau}$, then we define this quantity to be zero. We collect all these influences into a single \emph{influence matrix} $\infl_{\mu^{\tau}}$. It will also be convenient to introduce vertex-to-vertex influence submatrices $\infl_{\mu^{\tau}}^{r \to v} \in \R^{q \times q}$ given by
\begin{align}\label{eq:vtx-to-vtx-infl-submat}
    \infl_{\mu^{\tau}}^{r \to v}(\mfb,\mfc) \defeq \infl_{\mu^{\tau}}((r,\mfb) \to (v,\mfc)), \qquad \forall \mfb,\mfc \in [q].
\end{align}
\end{defn}
It is well-known that $\infl_{\mu^{\tau}}$ has real eigenvalues; see \cite{AL20,ALO21,CGSV21,FGYZ21}.

\begin{defn}[Spectral Independence]
We say the Gibbs distribution $\mu$ of a spin system on a finite $n$-vertex graph $G=(V,E)$ is \emph{spectrally independent with constant $\eta$} if for every $\Lambda \subseteq V$ with $|\Lambda| \leq n - 2$ and every valid pinning $\tau$ on $\Lambda$, we have the inequality $\eigval_{\max}(\infl_{\mu^{\tau}}) \leq \eta$, where $\eigval_{\max}$ denotes the maximum eigenvalue of a square matrix with real eigenvalues.
\end{defn}

It was established in \cite{CLV21a} that $O(1)$-spectral independence for the Gibbs distribution of a spin system on a bounded-degree graph implies optimal $O(n\log n)$ mixing time for Glauber dynamics. We now show that contraction of belief propagation in the sense of \cref{def:contraction}, combined with \cref{q:main}, implies $O(1)$-spectral independence. The following theorem mirrors a result of \cite{CLMM23}, which shows that contraction implies $O(1)$-spectral independence in the case of trees. It also generalizes strategies used previously to establish spectral independence via correlation decay and self-avoiding walk trees \cite{ALO21, CLV20, FGYZ21, CGSV21}.

\begin{theorem}
Fix a symmetric interaction matrix $A \in \R_{\geq0}^{q \times q}$. Suppose $A$ is up-to-$\Delta$ contracting with respect to a norm $\norm{\cdot}_{K}$ on $\R^{q}$ and a bounded potential $\varphi$, and further assume that \cref{q:main} has an affirmative answer for $A$ with respect to the potential $\varphi$. Then the Gibbs distribution $\mu$ of the spin system associated to $A$ on any graph of maximum degree $\Delta$ is $O_{A,\Delta,\varphi}(1/\delta)$-spectrally independent, where the $O_{A,\Delta,\varphi}(\cdot)$ hides constants depending on $A,\Delta$ and the boundedness of $\varphi$ (and $\delta$ is as in~\cref{def:contraction}). In particular, Glauber dynamics with respect to $\mu$ mixes in $O(n\log n)$-steps.
\end{theorem}
\begin{proof}
As was previously mentioned, $O(n\log n)$-mixing of Glauber dynamics follows from $O(1)$-spectral independence and \cite{CLV21a}. To establish $O(1)$-spectral independence, we suitably modify the strategy from \cite{CLMM23}, who showed that up-to-$\Delta$ contraction of belief propagation implies spectral independence in the special case of trees with maximum degree $\Delta$. Consider the hybrid vector norm $\norm{\bm{x}}_{\infty,K} = \max_{v \in V} \norm{\bm{x}_{v}}_{K}$ on $\R^{V \times [q]} \cong \wrapp{\R^{q}}^{V}$. We will establish
\begin{align*}
    \norm{\infl_{\mu^{\tau}}} \leq O_{A,\Delta,\varphi}(1/\delta), \qquad \forall \Lambda \subseteq V \text{ with } \abs{\Lambda} \leq \abs{V} - 2, \forall \tau : \Lambda \to [q],
\end{align*}
where the matrix norm $\|\cdot\|$ is induced by the hybrid vector norm $\norm{\cdot}_{\infty,K}$. This is enough since it is a standard linear algebraic fact that $\eigval_{\max}(M) \leq \norm{M}$ for any square matrix $M$ with real eigenvalues, where $\norm{\cdot}$ is any matrix norm induced by a vector norm. For cleanliness of notation, we will assume $\Lambda = \emptyset$; the same argument works mutatis mutandis when $\Lambda \neq \emptyset$. We will need the following lemma, which expresses the vertex-to-vertex influence submatrices as an average of Jacobians of the belief propagation functional. We provide a proof at the end of this section.
\begin{lemma}\label{lem:infl-decomp-jac}
Let $\mu$ be the Gibbs distribution of a $q$-spin system on a finite $n$-vertex graph $G=(V,E)$ with symmetric interaction matrix $A \in \R_{\geq0}^{q \times q}$. Let $r \in V$ be any vertex with neighbors $N(r) = \{u_{1},\dots,u_{d}\}$. If \cref{q:main} holds for $\xi$ with respect to a bounded potential $\varphi$ with $\Phi = \varphi'$, then
\begin{align*}
    \infl_{\mu_{G}}^{r \to v} = \E_{(\bm{\lambda}^{(1)},\dots,\bm{\lambda}^{(d)}) \sim \xi}\wrapb{\sum_{i=1}^{d} D_{r}^{-1} \cdot \wrapp{J_{i}F^{\circ \varphi}}\wrapp{\overrightarrow{\bm{m}}} \cdot D_{i} \cdot \infl_{\bm{\lambda}^{(i)} \star \mu_{G-r}}^{u_{i} \to v}},
\end{align*}
where
\begin{itemize}
    \item $\wrapp{J_{i}F^{\circ \varphi}}$ is the $i$th $q \times q$ submatrix of the Jacobian $JF^{\circ \varphi}$ corresponding to $u_{i}$,
    \item $D_{r} = \diag\wrapp{\mu_{G,r} \odot \Phi\wrapp{\mu_{G,r}}}$ is a diagonal matrix (where $\odot$ denotes entrywise multiplication),
    \item $\bm{p}_{i} = \wrapp{\bm{\lambda}^{(i)} \star \mu_{G-r}}_{u_{i}}$ and $\bm{m}_{i} = \varphi\wrapp{\bm{p}_{i}}$ for every $i=1,\dots,d$, and
    \item $D_{i} = \diag\wrapp{\bm{p}_{i} \odot \Phi\wrapp{\bm{p}_{i}}}$ is a diagonal matrix for each $i=1,\dots,d$.
\end{itemize}
(Note that throughout, to avoid cluttering the notation, we suppress the dependence of the diagonal matrices $D_{i}$ and the vectors $\bm{p}_{i},\bm{m}_{i}$ on the fields $\bm{\lambda}^{(i)}$.)
\end{lemma}

With this lemma in hand, we will bound the norm of the influence matrix $\infl_{\mu_{G}}$ by decomposing it using \cref{lem:infl-decomp-jac} and bounding the norms of various pieces in the decomposition. Since nested applications of \cref{lem:infl-decomp-jac} can quickly become notationally cumbersome, we introduce a bookkeeping gadget which can be viewed as a randomized analog of Weitz's self-avoiding walk tree in the $q = 2$ setting.

An \emph{instance} is a triple $\mathcal{I} = (S,v,\bm{\lambda})$, where $S \subseteq V$ is a subset of vertices, $v \in S$ is a distinguished vertex of the instance, and $\bm{\lambda}$ is a collection of external fields applied to a subset of vertices in $S$. The ``purpose'' of this instance is to ``compute'' the marginal distribution of $v$ with respect to the Gibbs distribution on the induced subgraph $G[S]$ after the external fields $\bm{\lambda}$ are applied. We write $\bm{p}_{\mathcal{I}} = \wrapp{\bm{\lambda} \star \mu_{G[S]}}_{v}$ for this distinguished marginal distribution. We also write $\mu_{\mathcal{I}} = \bm{\lambda} \star \mu_{G[S]}$ for convenience.

Now fix a root vertex $r \in V$ of $G$. \cref{q:main} naturally leads to a \emph{random computation tree $\mathcal{T}$} where the nodes of $\mathcal{T}$ are labeled by instances, and the root node is labeled by $\mathcal{I}^{(0)} = (V,r,\emptyset)$. The edges of this random tree $\mathcal{T}$ are determined as follows: A node labeled by an instance $\mathcal{I} = (S,v,\bm{\lambda})$ spawns $1 \leq d \leq \Delta - 1$ child nodes, each labeled by the instance $\wrapp{S - v, u_{i}, \bm{\lambda}^{(i)} \sqcup \bm{\lambda}}$ where $u_{i}$ is the $i$th neighbor of $v$ in $S$ (ordered arbitrarily), and we concatenate the external fields $\bm{\lambda}^{(i)},\bm{\lambda}$ (written as $\bm{\lambda}^{(i)} \sqcup \bm{\lambda}$). The new fields $\bm{\lambda}^{(1)},\dots,\bm{\lambda}^{(d)}$ are drawn randomly from $\xi$, and so these child instances are indeed random. Nonetheless, the resulting tree $\mathcal{T}$ has branching factor $\leq \Delta - 1$ with probability one, and if one only tracks the distinguished vertex $v$ in each instance of the tree, then one recovers the tree of self-avoiding walks of $G$. 

With this notation in hand, for each pair of vertices $r,v \in V$, we may decompose the vertex-to-vertex influence submatrix $\infl_{\mu_{G}}^{r \to v}$ as
\begin{align}\label{eq:infl-rv-decomp}
    \E_{\mathcal{T}}\wrapb{\sum_{k=0}^{\infty} \sum_{\substack{\mathcal{I}^{(0)} \to \dotsb \to \mathcal{I}^{(k)} \\ \text{distinguished vertex of } \mathcal{I}^{(k)} \text{ is } v}} \wrapp{\prod_{j=1}^{k} D_{\mathcal{I}^{(j-1)}}^{-1} \cdot \wrapp{J_{\mathcal{I}^{(j)}}F^{\circ \varphi}}\wrapp{\overrightarrow{\bm{m}}_{\mathcal{I}^{(j-1)}}} \cdot D_{\mathcal{I}^{(j)}}} \cdot \infl_{\mu_{\mathcal{I}^{(k)}}}^{v \to v}}
\end{align}
where $D_{\mathcal{I}} = \diag\wrapp{\bm{p}_{\mathcal{I}} \odot \Phi\wrapp{\bm{p}_{\mathcal{I}}}}$ is a diagonal matrix for any instance $\mathcal{I}$, $\overrightarrow{\bm{m}}_{\mathcal{I}^{(j-1)}}$ is the tuple of $\varphi\wrapp{\bm{p}_{\mathcal{I}}}$ over all child instances $\mathcal{I}$ of $\mathcal{I}^{(j-1)}$, and the Jacobian $J_{\mathcal{I}^{(j)}}F^{\circ \varphi}$ of course is with respect to the distinguished vertex of $\mathcal{I}^{(j)}$, which is a neighbor of the distinguished vertex of $\mathcal{I}^{(j-1)}$. Note that $\infl_{\mu_{\mathcal{I}^{(k)}}}^{v \to v}$ is a ``self-influence'' which is given by $I_{q} - \allone\bm{p}_{\mathcal{I}^{(k)}}^{\top}$.

With the decomposition \cref{eq:infl-rv-decomp} in hand, we now proceed to bound $\norm{\infl_{\mu_{G}}}$. Fix a vector $\bm{x} \in \R^{nq}$. In the following summations, if $\mathcal{I}^{(k)}$ is an instance, then we write $v^{(k)}$ for its corresponding distinguished vertex. We have
\begin{align*}
    &\norm{\infl_{\mu_{G}}\bm{x}}_{\infty,K} \\
    &= \max_{r \in V} \norm{\sum_{v \in V} \infl_{\mu_{G}}^{r \to v} \bm{x}_{v}}_{K} \\
    &= \max_{r \in V} \norm{\E_{\mathcal{T}}\wrapb{\sum_{k=0}^{\infty} \sum_{\substack{\mathcal{I}^{(0)} \to \dotsb \to \mathcal{I}^{(k)} \\ \text{in } \mathcal{T}}} \wrapp{\prod_{j=1}^{k} D_{\mathcal{I}^{(j-1)}}^{-1} \cdot \wrapp{J_{\mathcal{I}^{(j)}}F^{\circ \varphi}}\wrapp{\overrightarrow{\bm{m}}_{\mathcal{I}^{(j-1)}}} \cdot D_{\mathcal{I}^{(j)}}} \cdot \infl_{\mu_{\mathcal{I}^{(k)}}}^{v^{(k)} \to v^{(k)}} \cdot \bm{x}_{v^{(k)}}}}_{K} \tag{\cref{eq:infl-rv-decomp}} \\
    &\leq \sup_{\mathcal{T}} \max_{r \in V} \sum_{k=0}^{\infty} \norm{\sum_{\substack{\mathcal{I}^{(0)} \to \dotsb \to \mathcal{I}^{(k)} \\ \text{in } \mathcal{T}}} \wrapp{\prod_{j=1}^{k} D_{\mathcal{I}^{(j-1)}}^{-1} \cdot \wrapp{J_{\mathcal{I}^{(j)}}F^{\circ \varphi}}\wrapp{\overrightarrow{\bm{m}}_{\mathcal{I}^{(j-1)}}} \cdot D_{\mathcal{I}^{(j)}}} \cdot \infl_{\mu_{\mathcal{I}^{(k)}}}^{v^{(k)} \to v^{(k)}} \cdot \bm{x}_{v^{(k)}}}_{K} \tag{Convexity} \\
    &= \sup_{\mathcal{T}} \max_{r \in V} \sum_{k=0}^{\infty} \norm{D_{\mathcal{I}^{(0)}}^{-1} \sum_{\substack{\mathcal{I}^{(0)} \to \dotsb \to \mathcal{I}^{(k)} \\ \text{in } \mathcal{T}}} \wrapp{\prod_{j=1}^{k} \wrapp{J_{\mathcal{I}^{(j)}}F^{\circ \varphi}}\wrapp{\overrightarrow{\bm{m}}_{\mathcal{I}^{(j-1)}}}} D_{\mathcal{I}^{(k)}}\infl_{\mu_{\mathcal{I}^{(k)}}}^{v^{(k)} \to v^{(k)}}\bm{x}_{v^{(k)}}}_{K} \\
    &\leq O_{A,\Delta,\varphi}(1) \cdot \sup_{\mathcal{T}} \max_{r \in V} \sum_{k=0}^{\infty} \norm{\sum_{\substack{\mathcal{I}^{(0)} \to \dotsb \to \mathcal{I}^{(k)} \\ \text{in } \mathcal{T}}} \wrapp{\prod_{j=1}^{k} \wrapp{J_{\mathcal{I}^{(j)}}F^{\circ \varphi}}\wrapp{\overrightarrow{\bm{m}}_{\mathcal{I}^{(j-1)}}}} \bm{y}_{\mathcal{I}^{(k)}}}_{K}.
\end{align*}
In the last step, we used boundedness of $\varphi$ to pull out a multiplicative factor of $\norm{D_{\mathcal{I}^{(0)}}^{-1}} = O_{A,\Delta,\varphi}(1)$, where the matrix norm is induced by the vector norm $\norm{\cdot}_{K}$ on $\R^{q}$; we also write $$\bm{y}_{\mathcal{I}^{(k)}} = D_{\mathcal{I}^{(k)}}\infl_{\mu_{\mathcal{I}^{(k)}}}^{v^{(k)} \to v^{(k)}}\bm{x}_{v^{(k)}}.$$
Again by boundedness of $\varphi$, we have
\begin{align*}
    \norm{\bm{y}_{\mathcal{I}^{(k)}}}_{K} \leq \norm{D_{\mathcal{I}^{(k)}}\infl_{\mu_{\mathcal{I}^{(k)}}}^{v^{(k)} \to v^{(k)}}} \cdot \norm{\bm{x}_{\mathcal{I}^{(k)}}}_{K} \leq O_{A,\Delta,q}(1) \cdot \norm{\bm{x}_{\mathcal{I}^{(k)}}}_{K}.
\end{align*}
Hence, if we can show that for any $r \in V$, any realization of $\mathcal{T}$, and any $k \in \N$, the $k$th term of the summation is upper bounded by $(1 - \delta)^{k} \cdot \max_{\mathcal{I}^{(k)}} \norm{\bm{y}_{\mathcal{I}^{(k)}}}_{K}$, then we would be done. We will establish this by leveraging contraction. Observe that
\begin{align*}
    &\norm{\sum_{\substack{\mathcal{I}^{(0)} \to \dotsb \to \mathcal{I}^{(k)} \\ \text{in } \mathcal{T}}} \wrapp{\prod_{j=1}^{k} \wrapp{J_{\mathcal{I}^{(j)}}F^{\circ \varphi}}\wrapp{\overrightarrow{\bm{m}}_{\mathcal{I}^{(j-1)}}}} \bm{y}_{\mathcal{I}^{(k)}}}_{K} \\
    &\leq (1 - \delta) \cdot \max_{\mathcal{I}^{(1)}} \norm{\sum_{\substack{\mathcal{I}^{(1)} \to \dotsb \to \mathcal{I}^{(k)} \\ \text{in } \mathcal{T}}} \wrapp{\prod_{j=2}^{k} \wrapp{J_{\mathcal{I}^{(j)}}F^{\circ \varphi}}\wrapp{\overrightarrow{\bm{m}}_{\mathcal{I}^{(j-1)}}}} \bm{y}_{\mathcal{I}^{(k)}}}_{K} \tag{Contraction} \\
    &\leq \dotsb \tag{Induction} \\
    &\leq (1 - \delta)^{k} \cdot \max_{\mathcal{I}^{(k)}} \norm{\bm{y}_{\mathcal{I}^{(k)}}}_{K}.
\end{align*}
This completes the proof.
\end{proof}

\begin{proof}[Proof of \cref{lem:infl-decomp-jac}]
Again, we follow the proof in \cite{CLMM23}, which established an analogous identity in the case where the graph is a tree (so that there is no outer expectation with respect to $\xi$). We establish the claim analytically by expressing the marginal distribution as the derivative of the log-partition function, where the variables correspond to the external fields. We stress that the proof is actually \emph{agnostic} to the specific form of belief propagation; all that matters is that \cref{q:main} gives a way to exactly compute the marginals of a vertex $r$ given the marginals of its neighbors.

For variables $\bm{t} \in \R^{V \times [q]}$, define
\begin{align*}
    \mathcal{F}_{G}(\bm{t}) &\defeq \log \sum_{\sigma : V \to [q]} \wrapp{\prod_{uv \in E} A_{\sigma(u),\sigma(v)}} \cdot \exp\wrapp{\sum_{v \in V} \bm{t}_{v,\sigma(v)}} \\
    \mathcal{P}_{G,r,\mfc}(\bm{t}) &\defeq \frac{\partial}{\partial t_{r,\mfc}} \mathcal{F}_{G}(\bm{t}).
\end{align*}
Note that $\mathcal{F}_{G}(\bm{t})$ is the log-partition function for the tilted measure $e^{\bm{t}} \star \mu$, and $\mathcal{P}_{G,r,\mfc}(\bm{0}) = \mu_{G,r}(\mfc)$. Now observe that
\begin{align*}
    \varphi\wrapp{\mathcal{P}_{G,r}(\bm{t})} &= \E_{\wrapp{\bm{\lambda}^{(1)},\dots,\bm{\lambda}^{(d)}} \sim \xi}\wrapb{\varphi\wrapp{F_{A}\wrapp{\bigotimes_{i=1}^{d} \wrapp{\bm{\lambda}^{(i)} \star e^{\bm{t}} \star \mu}_{i}}}}, \qquad \forall \bm{t} \in \R^{V \times [q]}.
\end{align*}
Crucially, universality of the mixture $\xi$ means that \emph{$\xi$ does not depend on $\bm{t}$ whatsoever}. Hence, for each $\mfb,\mfc \in [q]$, if we differentiate entry $\mfc$ of the left-hand side with respect to the variable $t_{v,\mfb}$ for some $v \in V$ and $\mfb \in [q]$, we may push the differential inside the expectation by linearity. In particular,
\begin{align*}
    \frac{\partial}{\partial t_{v,\mfb}} \varphi\wrapp{\mathcal{P}_{G,r,\mfc}(\bm{t})} &= \E_{\wrapp{\bm{\lambda}^{(1)},\dots,\bm{\lambda}^{(d)}} \sim \xi}\wrapb{\frac{\partial}{\partial t_{v,\mfb}} F_{A,\mfc}^{\circ \varphi}\wrapp{\bm{m}_{1}\wrapp{\bm{\lambda}^{(1)},\bm{t}},\dots,\bm{m}_{d}\wrapp{\bm{\lambda}^{(d)},\bm{t}}}} \\
    &= \E_{\wrapp{\bm{\lambda}^{(1)},\dots,\bm{\lambda}^{(d)}} \sim \xi}\wrapb{\sum_{i=1}^{d} \sum_{\mfa \in [q]} \wrapp{\partial_{i,\mfa}F_{A,\mfc}^{\circ \varphi}}\wrapp{\overrightarrow{\bm{m}}} \cdot \frac{\partial}{\partial t_{v,\mfb}} \varphi\wrapp{\bm{p}_{i,\mfa}\wrapp{\bm{\lambda}^{(i)},\bm{t}}}}
\end{align*}
where $\bm{p}_{i,\mfa}\wrapp{\bm{\lambda}^{(i)},\bm{t}} = \wrapp{\bm{\lambda}^{(i)} \star e^{\bm{t}} \star \mu_{G-r}}_{u_{i}}(\mfa)$ and $\bm{m}_{i,\mfa}\wrapp{\bm{\lambda}^{(i)},\bm{t}} = \varphi\wrapp{\bm{p}_{i}\wrapp{\bm{\lambda}^{(i)},\bm{t}}_{\mfa}}$ for each $i \in [d]$ and $\mfa \in [q]$. On the other hand,
\begin{align*}
    \frac{\partial}{\partial t_{v,\mfb}} \log \mathcal{P}_{G,r,\mfc}(\bm{t}) \Big|_{\bm{t} = \bm{0}} = \infl_{\mu_{G}}^{r \to v}(\mfc,\mfb)
\end{align*}
and so simultaneously, we have
\begin{align*}
    \frac{\partial}{\partial t_{v,\mfb}} \varphi\wrapp{\mathcal{P}_{G,r,\mfc}(\bm{t})} \Big|_{\bm{t} = \bm{0}} &= \frac{\partial}{\partial t_{v,\mfb}} \wrapp{\varphi \circ \exp}\wrapp{\log \mathcal{P}_{G,r,\mfc}(\bm{t})} \Big|_{\bm{t} = \bm{0}} \\
    &= \Phi\wrapp{\mathcal{P}_{G,r,\mfc}(\bm{0})} \cdot \mathcal{P}_{G,r,\mfc}(\bm{0}) \cdot \frac{\partial}{\partial t_{v,\mfb}} \log \mathcal{P}_{G,r,\mfc}(\bm{t}) \Big|_{\bm{t} = \bm{0}} \\
    &= D_{r}(\mfc,\mfc) \cdot \infl_{\mu_{G}}^{r \to v}(\mfc,\mfb).
\end{align*}
By the same token,
\begin{align*}
    \frac{\partial}{\partial t_{v,\mfb}} \varphi\wrapp{\bm{p}_{i,\mfa}\wrapp{\bm{\lambda}^{(i)},\bm{t}}} \Big|_{\bm{t} = \bm{0}} = D_{i}(\mfa,\mfa) \cdot \infl_{\bm{\lambda}^{(i)} \star \mu_{G-r}}^{u_{i} \to v}(\mfa,\mfb).
\end{align*}
Substituting these calculations back in, aggregating them into a matrix over all $\mfb,\mfc \in [q]$, and writing everything as matrix-matrix products then finishes the proof.
\end{proof}

\section{Signature analysis of \texorpdfstring{$B$}{B}}
We show that our class of counterexamples, namely interaction matrices $B = B(\beta,A)$ as defined in~\cref{e:b} for $\beta \geq 1$ and $A$ satisfying the assumptions of \cref{t:technical}, includes all possible signatures apart from the antiferromagnetic signature (i.e. exactly one positive eigenvalue). Furthermore, no such interaction matrix can have the antiferromagnetic signature.

\begin{lemma}\label{lem:atleasttwo-poseig}
Fix $q \in \N$ with $q \geq 2$. Then for every $\beta \geq 1$ and every interaction matrix $A \in \R_{\geq0}^{q \times q}$ satisfying $A \geq \allone_{q}\allone_{q}^{\top}$ entrywise with $A \neq \allone_{q}\allone_{q}^{\top}$, the matrix $B = B(\beta,A)$ defined as in \cref{e:b} has at least two positive eigenvalues.
\end{lemma}
\begin{proof}
For $\gamma, \delta \in \RR$, note that
\begin{align*}
\left[
\begin{array}{c}
\gamma \\
\delta \ind_q
\end{array}
\right]^\top
B
\left[
\begin{array}{c}
\gamma \\
\delta \ind_q
\end{array}
\right]
&= \beta \gamma^2 + 2 \gamma \delta \langle \ind_q, \ind_q \rangle + \delta^2 \ind_q^{\top} A \ind_q \\
&> \gamma^2 + 2 \gamma \delta q + \delta^2 q^2 \tag{Using $\beta \ge 1, A \ge \ind_q \ind_q^{\top}, A \neq \ind_q \ind_q^{\top}$} \\
&= (\gamma + \delta q)^2 \\
&\geq 0.
\end{align*}
Thus $B$ is positive definite on the two-dimensional subspace spanned by the vectors $\allone_{\mfc^{*}}, \allone_{[q]} \in \R^{q+1}$. In particular, $B$ must have at least two strictly positive eigenvalues.
\end{proof}

\begin{proposition}\label{prop:signature}
For every $d,q \in \N$ with $d,q \geq 2$, every $\beta > 1$ and every integer $0 \leq k \leq q-1$, there exists an interaction matrix $A$ satisfying the conditions of~\cref{t:technical} such that $B(\beta,A)$ (defined as in \cref{e:b}) has exactly $k + 2$ positive eigenvalues and $q - k - 1$ negative eigenvalues.
\end{proposition}
\begin{proof}
Fix $d,q,\beta,k$ as in the claim. Consider the matrix
\begin{align*}
    A = \begin{bmatrix}
        \allone_{k}\allone_{k}^{\top} + (\beta - 1)\id_{k} & \allone_{k}\allone_{q-k}^{\top} \\
        \allone_{q-k}\allone_{k}^{\top} & t\wrapp{\allone_{q-k}\allone_{q-k}^{\top} + (\gamma - 1)\id_{q-k}}
    \end{bmatrix},
\end{align*}
where $0 < \gamma < 1$ is a parameter to be chosen later, and $t > 1$ is chosen to satisfy
\begin{align*}
    t^{d} = \frac{\beta^{d} + (q - k - 1)}{\gamma^{d} + (q - k - 1)}.
\end{align*}
Since $\beta > 1$, we have $t > 1$ even if we set $\gamma = 1$. Hence, by continuity, there exists $\epsilon = \epsilon(\beta,d) > 0$ such that for all $1 - \epsilon < \gamma < 1$, we have $t\gamma > 1$. These facts imply that $A \geq \allone_{q}\allone_{q}^{\top}$ with strict inequality for at least one entry. Furthermore, our choice of $t > 1$ implies that the function $\mfc \mapsto \sum_{\mfb \in [q]} A_{\mfc,\mfb}^{d}$ is constant, independent of $\mfc \in [q]$. Hence, $\mathscr{M}(d,A) = [q]$. Combining this with the fact that $A$ is not proportional to $\allone_{q}\allone_{q}^{\top}$ implies there exists some $\mfc \in [q]$ and $\mfa_{1},\mfa_{2} \in \mathscr{M}(d,A) = [q]$ such that $A_{\mfc,\mfa_{1}} \neq A_{\mfc,\mfa_{2}}$. This establishes that $A$ satisfies the conditions of~\cref{t:technical}.

Let us now compute the signature of
\begin{align*}
    B(\beta,A) = \begin{bmatrix}
        \allone_{k+1}\allone_{k+1}^{\top} + (\beta - 1)\id_{k+1} & \allone_{k+1}\allone_{q-k}^{\top} \\
        \allone_{q-k}\allone_{k+1}^{\top} & t\wrapp{\allone_{q-k}\allone_{q-k}^{\top} + (\gamma - 1)\id_{q-k}}
    \end{bmatrix}.
\end{align*}
Using the block determinant formula, its characteristic polynomial may be computed as
\begin{align*}
    &\det\wrapp{x\id_{q+1} - B(\beta,A)} \\
    &= \det\begin{bmatrix}
        (x - (\beta - 1))\id_{k+1} - \allone_{k+1}\allone_{k+1}^{\top} & -\allone_{k+1}\allone_{q-k}^{\top} \\
        -\allone_{q-k}\allone_{k+1}^{\top} & (x - t(\gamma - 1))\id_{q-k} - t\allone_{q-k}\allone_{q-k}^{\top}
    \end{bmatrix} \\
    &= \det\wrapp{(x - (\beta - 1))\id_{k+1} - \allone_{k+1}\allone_{k+1}^{\top}} \cdot \det\wrapp{(x - t(\gamma - 1))\id_{q-k} - (t + \delta)\allone_{q-k}\allone_{q-k}^{\top}}.
\end{align*}
where
\begin{align*}
    \delta \defeq \allone_{k+1}^{\top}\wrapp{(x - (\beta - 1))\id_{k+1} - \allone_{k+1}\allone_{k+1}^{\top}}^{-1}\allone_{k+1}.
\end{align*}
We have
\begin{align*}
    \det\wrapp{(x - (\beta - 1))\id_{k+1} - \allone_{k+1}\allone_{k+1}^{\top}} = \wrapp{x - (\beta - 1)}^{k} \cdot \wrapp{x - (\beta + k)},
\end{align*}
and
\begin{align*}
    \delta &= \allone_{k+1}^{\top}\wrapp{\frac{1}{x-(\beta-1)} \id_{k+1} + \frac{1}{\wrapp{x - (\beta - 1)}^{2}} \cdot \frac{1}{1 - \frac{k+1}{x - (\beta - 1)}} \cdot \allone_{k+1}\allone_{k+1}^{\top}}\allone_{k+1} \tag{Sherman--Morrison Formula} \\
    &= \frac{k+1}{x - (\beta - 1)} \wrapp{1 + \frac{k+1}{x - (\beta + k)}}.
\end{align*}
Hence, the second determinant term above equals
\begin{align*}
    \wrapp{x - t(\gamma - 1)}^{q-k-1} \cdot \wrapp{x - t(\gamma - 1) - (q-k)\wrapp{t + \frac{k+1}{x - (\beta - 1)} \wrapp{1 + \frac{k+1}{x - (\beta + k)}}}}.
\end{align*}
Putting together all these calculations, the full characteristic polynomial is given by
\begin{align*}
    \wrapp{x - (\beta - 1)}^{k-1}\wrapp{x - t(\gamma - 1)}^{q-k-1} \cdot q(x)
\end{align*}
where $q$ is the cubic polynomial
\begin{align*}
    q(x) &= \wrapp{x - t(\gamma - 1)}\wrapp{x - (\beta - 1)}\wrapp{x - (\beta + k)} \\
    &- (q - k)\wrapp{t\wrapp{x - (\beta - 1)}\wrapp{x - (\beta + k)} + (k+1)\wrapp{x - (\beta + k)} + (k+1)^{2}}.
\end{align*}
We will show that for $1 - \epsilon < \gamma < 1$ sufficiently close to $1$, $A$ has exactly $k + 2$ positive eigenvalues and $q - k - 1$ negative eigenvalues regardless of our choice of $t > 1$. The negative eigenvalues are already furnished by the $q - k - 1$ copies of $t(\gamma - 1)$, using our assumption that $0 < \gamma < 1$. We also have $k - 1$ copies of the eigenvalue $\beta - 1 > 0$. Hence, we must show that all $3$ roots of our degree-$3$ polynomial $q$ are all positive. Note that by the fact that $\allone_{k+1}\allone_{k+1}^{\top} + (\beta - 1)\id_{k+1}$ is positive definite, the Cauchy Interlacing Theorem implies that $B(\beta,A)$ has at least $k+1$ positive eigenvalues. In particular, we already know that $q$ has at least two positive roots. Hence, to show that all three roots of $q$ are positive, we just need to show that the constant term of $q$ is negative (since the leading coefficient of $q$ is positive). For this, observe that the constant term is given by
\begin{align*}
    &-t(\gamma - 1)(\beta - 1)(\beta + k) - (q-k)\wrapp{t(\beta - 1)(\beta + k) - (k+1)(\beta + k) + (k+1)^{2}} \\
    &= (q-k)(k+1)(\beta - 1) - t(\beta-1)(\beta + k)(q-k-1+\gamma).
\end{align*}
Since $0 \leq k \leq q - 1$, we have $q - k - 1 + \gamma > 0$ and $q - k > 0$. If we choose $1 - \epsilon < \gamma < 1$ sufficiently close to $1$ (depending on $\beta > 1$), we can guarantee that
\begin{align*}
    (q-k)(k+1) < (\beta + k)(q - k - 1 + \gamma).
\end{align*}
For such a $1 - \epsilon < \gamma < 1$, the constant term is strictly negative for every $t > 1$, and so we are done.
\end{proof}

\section{Verifying Numerical Inequalities}\label{sec:numerical}
\begin{proof}[Proof of \cref{claim:weird}]
Write $f(q,\beta) = \frac{q-1}{\log \frac{1}{\beta}} \log\frac{q-1}{(q - 1) - (1 - \beta)}$. First, we claim that for all $0 \leq \beta \leq 1$, $f$ is monotone decreasing in $q$ on $[1,\infty)$ with $\lim_{q \to \infty} f(q,\beta) = \frac{1 - \beta}{\log \frac{1}{\beta}}$. If we can show this, then we are done, since
\begin{align*}
    \frac{q-1}{\log \frac{1}{\beta}} \log\frac{q-1}{(q - 1) - (1 - \beta)} \geq \frac{1 - \beta}{\log \frac{1}{\beta}} \geq \sqrt{\beta},
\end{align*}
where in the final step we used the inequality $\log x \leq \frac{x - 1}{\sqrt{x}}$ for $x \geq 1$.

Now, observe that
\begin{align*}
    \frac{\partial}{\partial q} f(q,\beta) &= \frac{1}{\log \frac{1}{\beta}} \log \frac{1}{1 - \frac{1 - \beta}{q-1}} - \frac{q-1}{\log \frac{1}{\beta}} \cdot \frac{1-\beta}{(q-1)^{2} \cdot \wrapp{1 - \frac{1-\beta}{q-1}}} \\
    &= \frac{1}{\log \frac{1}{\beta}} \cdot \wrapp{\log \frac{1}{1 - \frac{1-\beta}{q-1}} + 1 - \frac{1}{1 - \frac{1-\beta}{q-1}}}.
\end{align*}
This is negative by applying the inequality $\log x \leq x - 1$ with $x = \frac{1}{1 - \frac{1-\beta}{q-1}}$, and by using $0 \leq \beta \leq 1$. To compute the limit, we apply L'H\^{o}pital's Rule to obtain
\begin{align*}
    \lim_{q \to \infty} f(q,\beta) = \frac{1}{\log \frac{1}{\beta}} \cdot \lim_{q \to \infty} \frac{1-\beta}{(q-1)^{2} \cdot \wrapp{1 - \frac{1-\beta}{q-1}}} \cdot \frac{1}{\frac{-1}{(q-1)^{2}}} = \frac{1-\beta}{\log \frac{1}{\beta}}
\end{align*}
as desired.
%The claim is immediate from the inequalities $\log \frac{1}{\beta} \leq \frac{1}{\beta} - 1$ and $\log\frac{q-1}{(q - 1) - (1 - \beta)} = \log \frac{1}{1 - \frac{1 - \beta}{q - 1}} \geq \frac{1-\beta}{q-1}$, which follow from $1 + x \leq e^{x} \leq \frac{1}{1 - x}$ on $[0,1]$, respectively.
\end{proof}
\begin{proof}[Proof of \cref{claim:insane-function}]
Write
\begin{align*}
    \Xi_{d,q,\epsilon}(\beta) &= (q-1) \frac{f(\beta)}{g(\beta)} \qquad \text{where} \\
    f(\beta) &= \frac{1}{d} \log\wrapp{\epsilon \beta^{-\frac{d}{10q}} + (1 - \epsilon)} - \log\wrapp{\epsilon \beta^{\frac{1}{10q(q-1)}} + (1 - \epsilon)} \\
    g(\beta) &= \log \frac{1}{\beta}.
\end{align*}
We will first show that $\Xi_{d,q,\epsilon}'(\beta) \leq 0$, or equivalently $f'(\beta)g(\beta) - f(\beta)g'(\beta) \leq 0$, on $\wrapb{1 - \frac{q}{d+1}, 1}$. To achieve this, we will show that $h(\beta) \defeq \beta\wrapp{f'(\beta)g(\beta) - f(\beta)g'(\beta)}$ evaluates to $0$ at $\beta = 1$, and that its derivative is nonnegative on $\wrapb{1 - \frac{q}{d+1}, 1}$. We have
\begin{align*}
    f'(\beta) &= -\frac{1}{d} \cdot \frac{\frac{d\epsilon}{10q} \beta^{-\frac{d}{10q}}}{\epsilon \beta^{-\frac{d}{10q}} + (1 - \epsilon)} \cdot \frac{1}{\beta} - \frac{\frac{\epsilon}{10q(q-1)}\beta^{\frac{1}{10q(q-1)}}}{\epsilon \beta^{\frac{1}{10q(q-1)}} + (1 - \epsilon)} \cdot \frac{1}{\beta} \qquad \text{and} \qquad g'(\beta) = -\frac{1}{\beta}.
\end{align*}
Hence
\begin{align*}
    h(\beta) = -\frac{1}{d} \cdot \frac{\frac{d\epsilon}{10q} \beta^{-\frac{d}{10q}}}{\epsilon \beta^{-\frac{d}{10q}} + (1 - \epsilon)} \cdot \log \frac{1}{\beta} - \frac{\frac{\epsilon}{10q(q-1)}\beta^{\frac{1}{10q(q-1)}}}{\epsilon \beta^{\frac{1}{10q(q-1)}} + (1 - \epsilon)} \cdot \log \frac{1}{\beta} + f(\beta).
\end{align*}
It is straightforward to check that $h(1) = 0$. Differentiating again, we obtain a miraculous cancellation, and arrive at
\begin{align*}
    h'(\beta) &= \underset{=0}{\underbrace{f'(\beta) + \frac{1}{\beta} \wrapp{\frac{1}{d} \cdot \frac{\frac{d\epsilon}{10q} \beta^{-\frac{d}{10q}}}{\epsilon \beta^{-\frac{d}{10q}} + (1 - \epsilon)} + \frac{\frac{\epsilon}{10q(q-1)}\beta^{\frac{1}{10q(q-1)}}}{\epsilon \beta^{\frac{1}{10q(q-1)}} + (1 - \epsilon)}}}} \\
    &- \log \frac{1}{\beta} \cdot \frac{\partial}{\partial \beta} \wrapp{\frac{1}{d} \cdot \frac{\frac{d\epsilon}{10q} \beta^{-\frac{d}{10q}}}{\epsilon \beta^{-\frac{d}{10q}} + (1 - \epsilon)} + \frac{\frac{\epsilon}{10q(q-1)}\beta^{\frac{1}{10q(q-1)}}}{\epsilon \beta^{\frac{1}{10q(q-1)}} + (1 - \epsilon)}}.
\end{align*}
In particular, we just need to verify that
\begin{align*}
    &\frac{1}{d} \cdot \frac{\frac{d\epsilon}{10q} \beta^{-\frac{d}{10q}}}{\epsilon \beta^{-\frac{d}{10q}} + (1 - \epsilon)} + \frac{\frac{\epsilon}{10q(q-1)}\beta^{\frac{1}{10q(q-1)}}}{\epsilon \beta^{\frac{1}{10q(q-1)}} + (1 - \epsilon)} \\
    &= \text{const.} - \frac{1-\epsilon}{10q} \cdot \Bigg(\underset{\defeq \ell(\beta)}{\underbrace{\frac{1}{\epsilon \beta^{-\frac{d}{10q}} + (1 - \epsilon)} + \frac{1}{q-1} \cdot \frac{1}{\epsilon \beta^{\frac{1}{10q(q-1)}} + (1 - \epsilon)}}}\Bigg).
\end{align*}
is a decreasing function in $\beta$ on $\wrapb{1 - \frac{q}{d+1}, 1}$, or equivalently that $\ell(\beta)$ is increasing on $\wrapb{1 - \frac{q}{d+1}, 1}$. For this, we have
\begin{align*}
    &\ell'(\beta) = \frac{\epsilon d}{10q} \cdot \frac{\beta^{-\frac{d}{10q}}}{\wrapp{\epsilon \beta^{-\frac{d}{10q}} + (1 - \epsilon)}^{2}} \cdot \frac{1}{\beta} - \frac{\epsilon}{10q(q-1)^{2}} \frac{\beta^{\frac{1}{10q(q-1)}}}{\wrapp{\epsilon \beta^{\frac{1}{10q(q-1)}} + (1 - \epsilon)}^{2}} \cdot \frac{1}{\beta} \geq 0 \\
    &\iff d(q-1)^{2} \geq \beta^{\frac{1}{10q(q-1)} + \frac{d}{10q}} \wrapp{\frac{\epsilon \beta^{-\frac{d}{10q}} + (1 - \epsilon)}{\epsilon \beta^{\frac{1}{10q(q-1)}} + (1 - \epsilon)}}^{2}.
\end{align*}
Since $1 - \frac{q}{d+1} \leq \beta \leq 1$, this is implied by the inequality
\begin{align*}
    d(q-1)^{2}(1-\epsilon)^{2} \geq \wrapp{\epsilon \wrapp{1 - \frac{q}{d+1}}^{-\frac{d}{10q}} + (1 - \epsilon)}^{2}.
\end{align*}
To see that this holds, observe that since $d \geq 2q$, we have
\begin{align}\label{eq:numer}
    \wrapp{1 - \frac{q}{d+1}}^{-\frac{d}{10q}} \leq \wrapp{\wrapp{1 - \frac{q}{d+1}}^{-\frac{d+1}{q}}}^{1/10} \leq \wrapp{\frac{3e}{2}}^{1/10} \leq 1.15.
\end{align}
Hence, the right-hand side above is upper bounded by
\begin{align*}
    \wrapp{1 + 0.15\epsilon}^{2} \leq 4(1 - \epsilon)^{2} \leq d(q-1)^{2}(1 - \epsilon)^{2},
\end{align*}
which holds since $q \geq 2, d \geq 2q \geq 4$, and $\epsilon \leq 1/4$.

Now, we bound $\Xi_{d,q,\epsilon}\wrapp{1 - \frac{q}{d+1}}$. We have
\begin{align*}
    \frac{q-1}{\log \frac{1}{\beta}} \leq \frac{q-1}{1 - \beta} \leq (d+1)
\end{align*}
where in the last step we used $\beta \geq 1- \frac{q}{d+1}$. Also have that for $\beta = 1 - \frac{q}{d+1}$
\begin{align*}
    \frac{\wrapp{\epsilon \beta^{-\frac{d}{10q}} + (1 - \epsilon)}^{1/d}}{\epsilon \beta^{\frac{1}{10q(q-1)}} + (1 - \epsilon)} &\leq \frac{\wrapp{1 + 0.15\epsilon}^{1/d}}{1 - \epsilon + \epsilon (3e/2)^{-\frac{1}{10(d+1)(q-1)}}} \tag{Using \cref{eq:numer}} \\
    &\leq \frac{\wrapp{1 + 0.15\epsilon}^{1/d}}{1 - \frac{\epsilon}{6(d+1)(q-1)}} \tag{Using $1.15^{-x} \geq 1 - \frac{x}{6}$ for $x \geq 0$} \\
    &\leq \wrapp{\frac{1 + 0.15\epsilon}{1 - \frac{\epsilon}{6(q-1)}}}^{1/d} \tag{Bernoulli Inequality} \\
    &\leq \wrapp{1 + \frac{\epsilon}{2}}^{1/d}.
\end{align*}
Combining these ingredients, we obtain
\begin{align*}
    \Xi_{d,q,\epsilon}\wrapp{1 - \frac{q}{d+1}} \leq \epsilon
\end{align*}
as desired.
\end{proof}

\end{document}